%% file: Supplementary_Material.tex
\documentclass[10pt,onecolumn,twoside,english]{IEEEtran}

\usepackage{amssymb}
\usepackage{amsmath}
\usepackage{graphicx,wrapfig,times,amsthm,bm,amsmath,amsfonts,amssymb}
\usepackage{subfig}
\usepackage{threeparttable}
\usepackage[wide]{sidecap}
\usepackage{array}
\usepackage{algorithm}
\usepackage{algpseudocode}
\usepackage{lipsum,babel}
\usepackage{color}
\usepackage{setspace}

\newtheorem{thm}{Theorem}
\newtheorem{lem}{Lemma}
\newtheorem{deff}{Definition}

\newtheorem{remark}{Remark}
\def \rank {\mathrm {rank}}
\def \as {\mathrm {~ ~ a.s.}}
\def \deg {\mathrm{deg}}
\def \Exp {\mathrm{Exp}}
\def \Poisson {\mathrm{Poisson}}

\begin{document}
\title{Optimality of Fast Matching Algorithms for Random Networks with Applications to Structural Controllability}

\author{Mohamad Kazem~Shirani Faradonbeh,
        Ambuj~Tewari,
        and~George~Michailidis% <-this % stops a space
}

\maketitle

\begin{abstract}
	Network control refers to a very large and diverse set of problems including controllability of linear time-invariant dynamical systems, where the objective is to select an appropriate input to steer the network to a desired state. There are many notions of controllability, one of them being {\em structural controllability}, which is intimately connected to finding maximum matchings on the underlying network topology. In this work, we study fast, scalable algorithms for finding maximum matchings for a large class of {\em random networks}. {First, we illustrate that {\em degree distribution} random networks are realistic models for real networks in terms of structural controllability.} Subsequently, we analyze a popular, fast and practical heuristic due to Karp and Sipser as well as a simplification of it. For both heuristics, we establish asymptotic optimality and provide results concerning the asymptotic size of maximum matchings for an extensive class of random networks. 
\end{abstract}

\input{introduction}

\input{random_networks}

\input{matching_algorithms}

\input{Numerical_Examples}

\input{matching_algorithms_in_random_networks}

\input{conclusion}

\input{appendix}

\bibliographystyle{IEEEtran}
\bibliography{References}

\end{document}

%% file: introduction.tex
\section{Introduction}

\IEEEPARstart{N}{etworks} are capable of capturing relationships between a set of entities (vertices) and have found applications in
diverse scientific fields including biology, engineering, economics and the social sciences \cite{Lin,Indika1,Sirkant}. 
Network control refers to a very large and diverse set of problems that involve control actions over a network (see for example, 
\cite{Dinits,Ford,Goldberg,Crabill,Tadj,Papadimitriou,Nino, Luenberger, Liu,Indika2} and references therein).

A class of control problems involves dynamical systems evolving over time that have inputs and outputs and many results exist for systems that
exhibit linear and time-invariant dynamics \cite{Luenberger}. One particular notion of control is that of {\em structural controllability} \cite{Lin,dion2003generic}, which
was recently explored by Liu et al. \cite{Liu}. Under this notion, the structural controllability problem reduces to find maximum matchings on appropriate matrices as reviewed in
Section \ref{SCMM}. The problem of obtaining maximum matchings has been extensively studied in the computer science literature both for deterministic
\cite{gibbons1985algorithmic} as well as random networks
\cite{bollobas2001random}. {However, the focus in the literature has been on special
classes of {\em undirected random networks} \cite{Frieze1,Frieze2}. \textcolor{black}{For example, using the results of Tao and Vu \cite{tao2014random}, one can infer that for \emph{dense} Erdos-Renyi random networks, a single controller is sufficient to ensure controllability.} Little is known about the performance of matching algorithms in other interesting classes of random networks \cite{Bohman}, with an exception of the recent work of Balister and Gerke \cite{balister2015controllability}. We focus on {\em degree distribution} random network models, and establish results on the minimum number of controllers needed to guarantee structural controllability. We also show that these models are realistic representations for real world applications.}

A popular, fast and practical algorithm for matchings on undirected random networks is due to Karp and Sipser \cite{KS}, which represents the cornerstone of our theoretical investigations and through it we provide generalizations of previous work in the literature to broader classes of undirected random networks. Further, we also extend the results for directed variants of the same classes of random networks.

\subsection{Structural Controllability and Maximum Matchings}\label{SCMM}
Next, we review some key concepts in structural controllability for linear dynamical systems. 
Consider a system described by a $n$-dimensional {\em state} vector $x(t)=(x_1(t), ..., x_n(t))^T \in \mathbb R ^n$, whose
dynamical evolution is described by 
$$ \frac{dx}{dt} = Ax(t)+ B u(t),$$
where $A \in \mathbb R ^ {n \times n}$ is the system transition matrix, $u(t)=(u_1(t), ..., u_k(t))^T \in \mathbb R ^k$  
captures control actions and $B \in \mathbb R ^ {n \times k}$ is the input matrix. Assuming that the $n$-dimensional system can be represented by
vertices on a network $G=(V,E)$ with $V= \{1, 2 , ..., n \}$ denoting the set of vertices and $E \subset V \times V$ the set of edges, it can be seen that the non-zero entries
in the transition matrix $A$ correspond to the directed edges in $E$. {Indeed, for $i,j \in V$ the edge $i \to j$ exists if and only if vertex $i$ influences vertex $j$, i.e. $A_{ji} \neq 0$.} Such a system is called {\em controllable} if for any initial state $x(0)=x_0$ and any 
desired state $x_d$ for some $T<\infty $ one can find an input matrix $B$ and control vectors $\{u(t)\}_{0 \leq t \leq T}$ so that the system reaches state $x_d$ i.e. $x(T)=x_d$. The minimum $k$ for which system can be controllable is called the minimum number of controllers.

The magnitude of the entries in the transition matrix $A$ captures the interaction strength between the vertices in the network; for example, the traffic on individual communication links in a communications network or the strength of a regulatory interaction in a biological network. {The time invariant matrix $B$ 
indicates which vertices are controlled by an outside controller. Hence, the 
set of vertices that when applying controllers to them makes the system controllable needs to be identified.} 

{Note that we don't assume any constraint on the number of nonzero elements in the columns of input matrix $B$, i.e. one controller $u_i(t)$ can influence multiple vertices. The case where every controller can be applied to only one vertex in the network (which in turn implies that the goal is to find the minimum number of vertices the controllers can be applied to) is studied by Olshevsky \cite{olshevsky2014minimal}. He shows that finding the exact solution to the problem is NP hard.}

The algebraic criterion to check controllability of a time invariant linear dynamical system is Kalman's controllability rank condition, that states  that controllability can be
achieved, if and only if the matrix $C=[B,AB, ..., A^{n-1}B]$ is full rank; i.e. $\rank(C)=n$. Note that $C \in \mathbb R ^ {n \times nk}$.
This algebraic criterion is computationally hard to check, especially for large systems. Further, in many applications, obtaining exact values of $A$ may
not be feasible and hence a tractable alternative is needed.

Thus, we say that a time invariant linear dynamical system is \emph{structurally controllable}, if it is possible to select the non-zero values of $A,B$, so that Kalman's rank condition is satisfied \cite{Lin}. A structurally controllable system is controllable for almost all $A,B$; i.e. the pathological cases for which a structurally controllable network is not controllable has zero Lebesgue measure. {The relation between the minimum number of controllers needed to structurally control a network and the size of its maximum matching has been presented in several forms (see \cite{commault2002characterization}). The version used in the current work is the one appearing in Liu et al. \cite{Liu} ``minimum inputs theorem" (stated below)}. \textcolor{black}{Moreover, similar equivalence between structural controllability and maximum matching when every controller can be applied to only one vertex, is studied by Assadi et al. \cite{assadi2015complexity}.}

\textcolor{black}{Furthermore, Commault and Dion \cite{commault2015single} study the problem of using only a single controller applied to as few vertices as possible. More general results regarding the control configuration selection can be found in the work of Pequito et al. \cite{pequito2013framework} which discusses relations to maximum matching problem as well. Other considerations must be taken into account for deciding which approach is most suitable for the practical application under consideration. Finally, we have no assumption regarding self-loops in the network. Cowan et al. \cite{cowan2012nodal} study networks, where every vertex influences itself.}

Algorithms to find a maximum matching are well studied in the 
literature and exhibit polynomial time complexity (with respect to the size of the network). A popular one developed by Micali and Vazirani \cite{MM} has running time 
$O \left ( |V|^{0.5} |E| \right )$.

Next, for completeness we provide a definition of maximum matching and also state the minimum inputs theorem.
\begin{deff} For a directed network $G=(V,E)$, a subset of edges $M$ is a matching, if no two edges in $M$ share a common starting or a common ending vertex. A maximum matching corresponds to a matching of maximum size.
\end{deff} 
{\begin{deff} \label{unmatched} Given a matching $M$ for directed network $G=(V,E)$, a vertex is matched, if it is an ending vertex of an edge in the matching $M$. Otherwise, it is unmatched.
\end{deff}}
{\bf Minimum Inputs Theorem}~\cite{Liu}: 
\emph{Let $M^*$ be a maximum matching of the network $G=(V,E)$. The minimum number of controllers needed for structural controllability of the network is $\max \{ 1 , n- |M^*| \}$. Moreover, for any matching $M$, the network is struturally controllable using $\max \{ 1 , n- |M| \}$ controllers.}

The upshot of this result is that in order to find the minimum number of required controllers for structural controllability we can equivalently find the size of a maximum matching. {Furthermore, one can explicitly find the structure of the input matrix $B$ according to the proof of minimum inputs theorem. \textcolor{black}{In fact, one can see that a matching is formed of a set of directed loops and a set of directed paths. The first vertex of each path is unmatched, while all other vertices are matched. According to Lin's work \cite{Lin}, one controller is used to actuate every unmatched vertex. Amongst matched vertices, some might need to be actuated by a controller, but no new controller is needed, as it suffices to apply any of the previously used controllers to one arbitrary vertex of each loop.} So the number of nonzero components in $B \in \mathbb R^{n \times k}$, or equivalently the total number of the connections between the inputs and vertices in the network, is at most $n$. More details are provided by Liu et al. (Supplementary Information of \cite{Liu}).}

{As the problem of finding the minimum number of controllers needed for structural controllability of a network  reduces to maximum matching type of problems, henceforth we use the (minimum) number of unmatched vertices and the (minimum) number of controllers interchangeably.} In this work, we provide results about the size of matchings obtained by different fast algorithms for classes of random networks.
{In fact, Karp and Sipser \cite{KS} proved that, for the classical undirected Erdos-Renyi random network, the $KS$ algorithm is optimal. We generalize their results to a larger class of random networks.}
The remainder of the paper is organized as follows. In Section \ref{RN} we introduce different classes of random networks subsequently studied in this work.
Furthermore, some probabilistic results needed for technical developments are summarized. The main algorithms studied are introduced in Section \ref{algorithmsRN} and {connections to real networks using some numerical examples are provided in Section \ref{numerical}. The key results of the paper are presented in Sections \ref{matching} (analysis and optimality of the algorithms).} 

{{\bf Notations:} For set $S$, let $|S|$ be the number of elements in $S$. ${n \choose n_1, \ldots, n_r} = \frac{n!}{n_1! \ldots n_r !}$. For vertices $u,v \in V$ in the network $G=(V,E)$, $\{ u,v \}$ ($(u,v)$ or $(v,u)$ in directed or bipartite networks) denotes an edge, so $N=M \cup \{ \{ u,v \} \}$ means adding the edge $\{ u,v \}$ to $M$ gives $N$. $G - \{ u,v \}$ means removing vertices $u,v$ (and so all edges connected to them) from network $G$. Further, $\deg_G (v) $ denotes the degree of $v$ in $G$ i.e. the number of edges in $G$ connected to $v$. When there is no subscript, the network $G$ is identifiable from the context. Finally, for $x \in \mathbb R ^n$, $n \in \{ 0,1,2 ,... \} \cup \{\infty\}$, $\| x \|_1$ is $\ell_1$ norm of $x$: $\| x \|_1= \sum \limits_{i=1}^n |x_i|$. }

%% file: random_networks.tex
\section{Random Networks}\label{RN}

In this section, we introduce different classes of random network models and then present some general results about convergence and concentration of real-valued functions on networks. In order to have a general framework which includes both directed and undirected networks, we note that every undirected network $G=(V,E)$ can be considered as a directed network in which, for all vertices $i,j$, both edges $i \to j$ and $j \to i$ exist, if and only if the edge $i \leftrightarrow j$ exists in the original undirected network. All statements presented are true for both directed and undirected networks unless explicitly mentioned. For a comprehensive discussion on constructions and properties of (undirected) random networks, see chapter 3 in Durrett \cite{DurretBook}.

%\subsection{Random Models for Networks} \label{randommodels}
The first model for random networks we consider is the {\bf Erdos-Renyi} $(ER)$ model. A directed network $G=(V,E)$, $V= \{ 1, 2, ..., n \}$ is (drawn from) $ER$ if every edge $i \rightarrow j$ for $ i,j=1, 2, ..., n$ is present in the network independently with probability $p^{(n)}$. Analogously, an undirected network $G=(V,E)$, $V= \{ 1, 2, ..., n \}$ is $ER$ if every edge $i \leftrightarrow j$ for $ i,j=1, 2, ..., n, i \leq  j$ is present in the network independently with probability $p^{(n)}$. Henceforth, for $\lambda \in [0,\infty]$, $ER(\lambda)$ is a Erdos-Renyi random network, for which $np^{(n)} \to \lambda$ as $n \to \infty$. In $ER(\lambda)$ random networks, the parameter $\lambda$ corresponds to the average degree. 

The next model we consider is the {\bf Uniform Fixed-Size} $(UFS)$ model. A directed network $G=(V,E)$, $V= \{ 1, 2, ..., n \}$ is $UFS$ when the cardinality of the edge set $|E| =k_n$ for some fixed $k_n$, and the $k_n$ directed edges are drawn uniformly among all $n^2$ possible edges. The construction for the undirected network is
similar, but the $k_n$ edges are chosen uniformly among all $\frac{n(n+1)}{2}$ possible edges. For $\lambda \in [0,\infty]$, we denote by $UFS(\lambda)$ a random network of the UFS class, for which $\frac{k_n}{n} \to \lambda$ for directed and $\frac{k_n}{n} \to \frac{\lambda}{2}$ for the undirected case, as $n \to \infty$. Once again, the $\lambda$ parameter corresponds to the average degree.

Finally, we introduce the class of {\bf Degree Distribution} $(DD)$ random networks. There are a couple of reasons for considering this class. First, it lets us consider networks with degree distributions commonly found in real networks that simpler models such as $ER$ (where the degree distribution is Poisson) cannot model. {For example, as shown in Fig. \ref{DDHisto1}, the degree distribution of real world networks can be any arbitrary non-parametric distribution.

Histograms demonstrating input and output degree distributions of some other real world networks can be found in Fig. \ref{DDHisto9}.
\begin{figure}[t!] 
        \centering
        \scalebox{.5}
        {\includegraphics {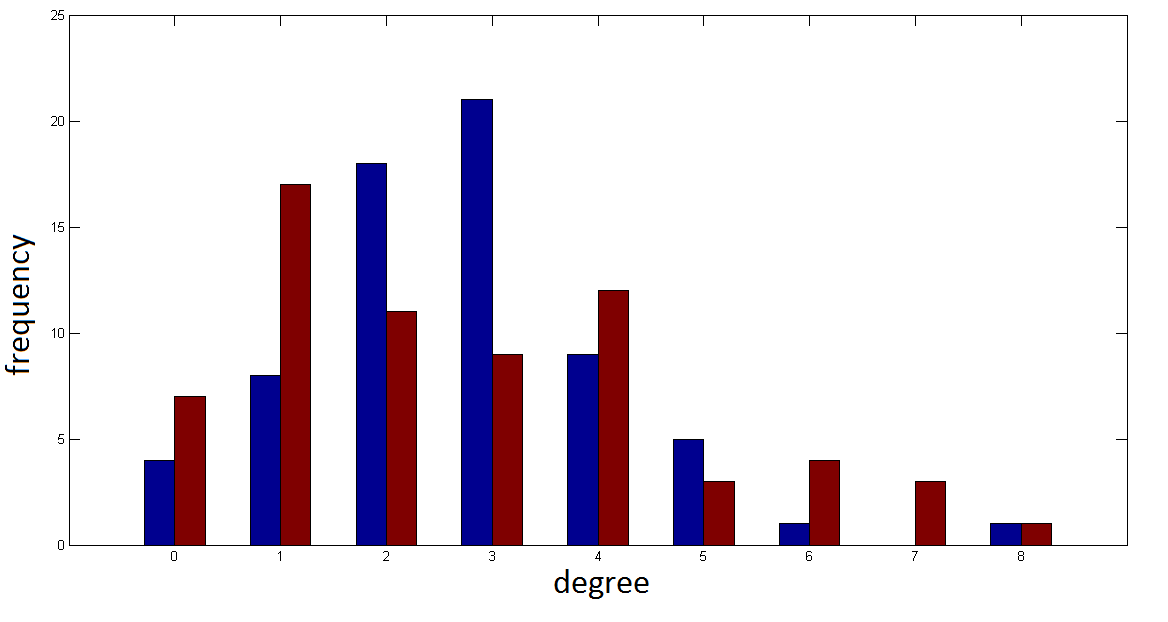}} 
	\caption{{Degree distribution histogram of the social network between prison inmates. For more information see Table \ref{nwtable}. (blue: input degree, red: output degree). As shown here, it is difficult to model degree distributions of real networks using standard parameteric distributions.} }
	\label{DDHisto1}
\end{figure}
\begin{figure}[t!] 
        \centering
        \scalebox{.65}
        {\includegraphics {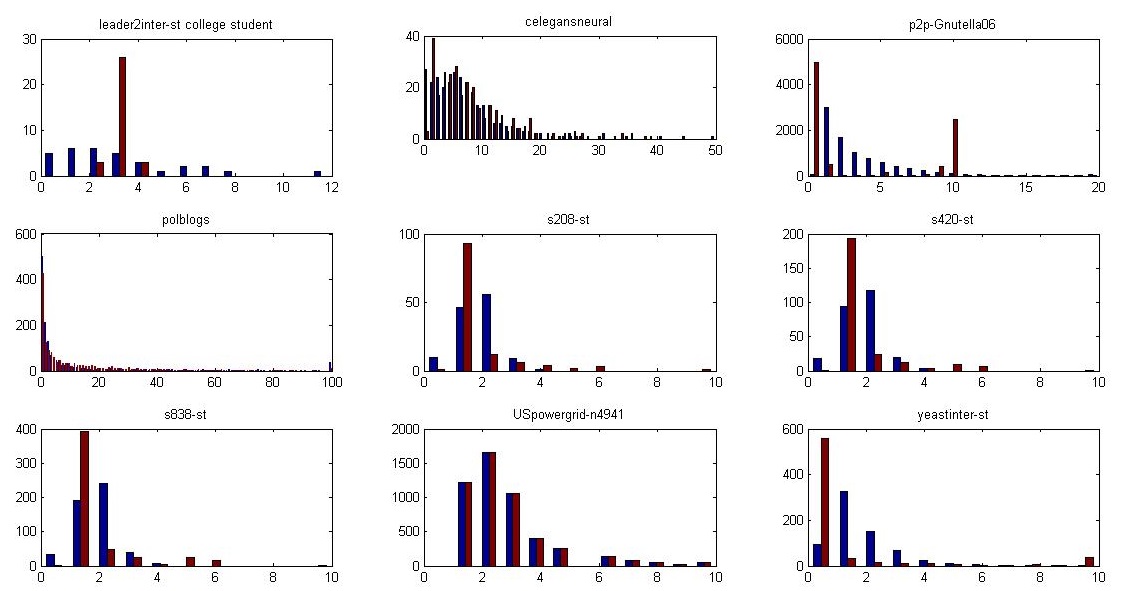}} 
	\caption{{Degree distribution histogram of some networks. For more information see Table \ref{nwtable}. (blue: input degree, red: output degree).} }
	\label{DDHisto9}
\end{figure}
 
{Second, as empirically shown in Section \ref{numerical}, the number of controllers for structural controllability of a network is, to a large extent, determined by its degree distribution. More details can be found in Section \ref{numerical} as well as the work of Liu et al. \cite{Liu}.} 

An undirected random network is a member of the $DD$ class, if for a given degree distribution the attachment of edges is random. Specifically, let $p$ be a probability distribution with support on the set $\{0,1,2, \ldots\}$ of nonnegative integers. We then construct an undirected network $DD(p)=(V,E)$ as follows. {Let $V= \{ 1, 2, ..., n \}$ and for $i \in V$, let vertex $i$ have $D_i$ undirected half-edge(s) (one-half of an edge is connected to vertex $i$). Note that $D_1, \ldots , D_n$ are the corresponding degrees which are assumed to be independent and identically distributed (iid) with distribution $p$; $\mathbb{P} (D_i = k) = p(k)$.} 
To complete the construction, we then pair all half-edges randomly; i.e. all $\sum \limits_{i=1}^n D_i \choose 2, \ldots , 2$ possible attachments of half-edges have equal probability.

When the number of half-edges $\sum \limits_{i=1}^n D_i$ is an even number, the construction is straightforward and the number of edges will be $\frac{1}{2} \sum \limits_{i=1}^n D_i$. When it is an odd number, we pair the half-edges randomly to obtain the network and omit the last single half-edge for which no pairing was established at the end of the construction, so that the number of edges will be $\frac{1}{2} (\sum \limits_{i=1}^n D_i -1)$. Note that the omission or presence of  multiple edges will lead to a difference between $D_1, \ldots , D_n$ and the actual observed degrees $\deg(1), \ldots , \deg(n)$ once the network construction is completed. However, as Lemma~\ref{RNDegDist} below establishes, the asymptotic empirical degree distribution will be the original degree distribution from which the network was constructed, as long as the expected value of $D_i$, $\mathbb E D_i$ is finite.

Viewing an undirected network $DD(p)$ as a directed one, both input and output degrees of vertex $i$ are $\deg(i)$. To construct a directed DD random network, denoted by $DD(p_{in},p_{out})$, with distinct input ($\deg_{in}$) and output ($\deg_{out}$) degrees, we do the following: once we have iid draws $D^{(in)}_i$ and $D^{(out)}_i$ from the input and output degree probability distributions $p_{in}$ and $p_{out}$ respectively:
$$\mathbb P (D^{(in)}_i = k) = p_{in}(k), \ \mathbb P (D^{(out)}_i = k) = p_{out}(k) ,$$
let vertex $i \in V$ have $D^{(in)}_i$ directed half-edges pointing into vertex $i$ and $D^{(out)}_i$ directed half-edges pointing out from vertex $i$. Next, we pair directed half-edges randomly to have $\min \{ \sum \limits_{i=0}^n D^{(in)}_i, \sum \limits_{i=0}^n D^{(out)}_i\}$ edges and omit the remaining half-edges. The random pairing of half-edges implies that all $\left ( \max \big\{ \sum \limits_{i=0}^n D^{(in)}_i, \sum \limits_{i=0}^n D^{(out)}_i\big\} \right ) !$ possible pairings of half-edges are equally likely. Note that $D_i^{(in)}, D_i^{(out)}$ do not need to be independent.

Furthermore, in general, the degrees do not need to be iid. In fact, as shown later in the paper, the key asymptotic results we establish are based on the empirical degree distributions which are, by the following lemma, same as the original degree distributions when vertex degrees are iid. However, as long as for all $k=0,1, \ldots$, $\lim \limits_{n \to \infty} \frac{| \{ i \in V : D_i=k \} |}{n}$ (equivalently $\lim \limits_{n \to \infty} \frac{| \{ i \in V : D^{(in)}_i=k \} |}{n}, \lim \limits_{n \to \infty} \frac{| \{ i \in V : D^{(out)}_i=k \} |}{n}$) are deterministic, our results hold using the resulting asymptotic empirical degree distributions.

\begin{lem} \label{RNDegDist}
For an undirected network $G=(V,E)$, define the asymptotic empirical degree distribution as 
$$\hat p (k) =\lim \limits_{n \to \infty} \frac{| \{ i \in V : \deg(i)=k \} |}{n}.$$
If $G=DD(p)$ is a random network and $\mu=\sum \limits_{k=0}^\infty kp(k) < \infty$, then the limit above exists and we have $\hat p (k)= p(k)  $ for all $k=1, 2, \ldots$. 
In general for a network $G=(V,E)$, define the asymptotic empirical input and output degree distributions as 
$$\hat p_{in}(k) =\lim \limits_{n \to \infty} \frac{| \{ i \in V : \deg_{in}(i)=k \} |}{n} $$
$$\hat p_{out}(k)= \lim \limits_{n \to \infty} \frac{| \{ i \in V : \deg_{out}(i)=k \} |}{n}.$$
If $G=DD(p_{in},p_{out})$ is a random network and $ \mu= \sum \limits_{k=0}^\infty kp_{in}(k) = \sum \limits_{k=0}^\infty kp_{out}(k) < \infty$,
then for all $k=1, 2, \ldots$ the limits above exist and we have $\hat p_{in} (k)= p_{in}(k), \hat p_{out}(k)=p_{out}(k) $.
\end{lem}
Henceforth, for all networks by $p_{in}, p_{out}$ we mean asymptotic empirical degree distributions $\hat p_{in}, \hat p_{out}$ respectively.

%%%%%%%%%%%%%%%%%%%%%%%%%%%%%%%%%%%%%%%%%%%%%%%%%%%%%%%%%%
%\subsection{Lipschitz Property and Martingale Difference Sequences}\label{MDS}

Next, we define the Lipschitz property for real valued functions defined over networks.
\begin{deff}\label{lipschitz} 
Let $f$ be a real-valued function on the set of directed networks. We say $f$ has the {\bf Lipschitz} property, if $| f(G_1)-f(G_2) | \leq 1$ whenever $G_1=(V,E_1), G_2=(V,E_2)$, $E_2=E_1 \cup \{ e \}$; i.e. the value of $f$ will change at most by 1 if one new edge is added to the network.
\end{deff}

%%%%%%%%%%%%%%%%%%%%%%%%%%%%%%%%%%%%%%%%%%%%%%%%%%%%%%%%%%
%\subsection{Convergence of Lipschitz Functions on Random Networks} \label{conv}
{\begin{remark}
Properly defining the norm $\| \cdot \|$ on the networks, one can see the Lipschitz property is the same as the classic notion $\left | f(G_1) - f(G_2) \right | \leq \| G_1 - G_2 \|$. 
Indeed, for every network $G=(V,E)$, letting $\mathcal A (G)$ be the adjacency matrix of $G$ ($\mathcal A (G) \in \mathbb R ^{|V| \times |V|}$, for $i,j \in V$ if $i \to j$, $\mathcal A (G)_{ji}=1$ and $\mathcal A (G)_{ji}=0$ otherwise), define $\| G \| = \| \mathcal A (G) \|_1 = \sum \limits_{i,j=1}^{|V|} | \mathcal A (G)_{ij}|$. Then $G_1, G_2$ differ by only one edge if and only if $\| \mathcal A (G_1) - \mathcal A (G_2) \|_1 =1$. So $f$ has the Lipschitz property if $| f(G_1) - f(G_2)|  = | f \left ( \mathcal A (G_1) \right ) - f \left ( \mathcal A (G_2) \right ) | \leq \| \mathcal A (G_1) - \mathcal A (G_2) \|_1$.
\end{remark}}
{Next, we present convergence and concentration inequalities for functions of random networks which have the Lipschitz property. Specifically,  the number of unmatched vertices (or equivalently the number of controllers) obtained by a matching algorithm, has the Lipschitz property as shown in Section \ref{matching}. The consequence of the Lipschitz property for a function defined on a ``not too dense" random network, is that it concentrates around its expected value. Further, if the average degree of the network is finite, then the concentration occurs exponentially fast.} These results are summarized in the following theorem.

\begin{thm}[Convergence of real-valued functions for random networks]\label{conv}
Let $G=(V,E)$ belong to the $ER$, $UFS$, $DD(p_{in},p_{out})$ or $DD(p)$ class. For a real-valued function $f$ which has the Lipschitz property, if
\begin{eqnarray*}
\label{DD1}
\limsup \limits_{n \rightarrow \infty} \frac{\mathbb E \left ( | E | \right ) }{n^2}  = 0,
\end{eqnarray*}
then $\frac{ f(G) - \mathbb{E}(f(G))}{n} \rightarrow _P 0$ as $n \to \infty$. If in addition,
\begin{eqnarray*}
\label{DD2}
\limsup \limits_{n \rightarrow \infty} \frac{\mathbb E \left ( | E | \right ) \log n}{n^2}  = 0,
\end{eqnarray*}
then $\frac{f(G) - \mathbb{E}(f(G))}{n} \rightarrow 0 \as$ as $n \to \infty$. When $\sup \limits_{n \geq 1} \frac{\mathbb E \left ( | E | \right ) }{n}  < \infty$ the rate of convergence is exponential; i.e. there is $C >0$ such that for every $0 < \epsilon < 1$ :
$$ \mathbb{P}(  \frac{| f(G) - \mathbb{E}(f(G)) |}{n} > \epsilon ) \leq 2 \exp (- n C \epsilon ^2 ).$$
\end{thm}

%% file: matching_algorithms.tex
\section{ Matching Algorithms} \label{algorithmsRN}

Before studying the algorithms, we provide a description of networks which will be useful later. As mentioned before, every undirected network can be considered as a directed one. Now to have a better understanding of how the algorithms work we view every directed network as a bipartite network $G=(L,R,E)$ where $L=R=V$, $E \subset L \times R$, {$L$ and $R$ are respectively the left and the right side of the bipartite network}, and for $l \in L, r \in R$ there is an edge $(l,r) \in E$ if and only if in the original directed network there is an edge from $l$ to $r$: $l \to r$. Henceforth we will only deal with bipartite networks.

Matching algorithms take a network as input and produce a matching as output. Maximum matching algorithms will give a matching of maximum size. Algorithm~\ref{greedyalgo} is the well known Greedy Algorithm that produces a suboptimal matching $M_G$ in general. {For an arbitrarily chosen vertex $v$ on the right side, Greedy tries to find an arbitrary vertex $u$ on the left side which is connected to $v$ and has not been used before by any other vertex on the right side. Note that as mentioned above, networks can be assumed to be bipartite.} 
{Moreover, clearly the time complexity of Greedy is $O(n)$.}
\begin{algorithm}
  \caption{{\bf: Greedy} \\ Input: $G=(L,R,E)$ \\ Output: matching $M_G(G)$}\label{greedyalgo}
  \begin{algorithmic}
      \State $M_G \gets \emptyset$
      \While{$E \neq \emptyset$}
         \State let $v \in R$
         \If{$\deg(v)=0$} 
            \State $G \gets G - \{ v \}$
          \ElsIf{for $u \in L$, $(u,v) \in E$}
             \State $G \gets G - \{ u, v \}$
             \State $M_G \gets M_G \cup \{ (u,v) \}$ 
           \EndIf
      \EndWhile
      \State \textbf{return} $M_G$
  \end{algorithmic}
\end{algorithm}

Note that Greedy picks an arbitrary vertex $v \in R$ in every iteration. Because the goal is to find a matching of largest possible size, this strategy for picking a vertex can be improved. First note that for every vertex $v \in R$ of degree one, there is a matching of maximum size in which $v$ is matched. The logic is as follows. Let $u \in L$ be the vertex on the left side connected to $v$, $(u,v) \in E$. There must exist a vertex on the right side, $w \in R$, such that $(u,w) \in E$, such that $w$ is matched to $u$ by a matching $M$ of maximum size; since, if not, adding $(u,w)$ to it leads to a matching of larger size. Now defining a new matching $M'$ which is exactly $M$ with $(u,w)$ removed and $(u,v)$ added, i.e. $M'= M  \cup \{ (u,v) \} - \{ (u,w) \}$, we have $|M|=|M'|$ i.e. $M'$ is a maximum matching as well. Hence as long as we can find a vertex of degree one, we can find a matching of exactly maximum size. In other words: \emph{no mistake occurs as long as a degree one vertex is picked in every iteration} of Greedy {(a mistake occurs if in an iteration, the algorithm adds an edge to the matching which is not optimal, i.e. leads to a deviation from maximum matching)}. 

\begin{algorithm}
  \caption{{\bf: Karp-Sipser} \\ Input: $G=(L,R,E)$ \\ Output: matching $M_{KS}(G)$}\label{ksalgo}
  \begin{algorithmic}
      \State $M_{KS} \gets \emptyset$
      \While{$E \neq \emptyset$}
         \State let $v=\arg\!\min \limits_{w \in L \cup R} \deg(w)$
         \If{$\deg(v)=0$} 
            \State $G \gets G - \{ v \}$
          \ElsIf{for $u \in L \cup R$, $\{u,v \} \in E$}
             \State $G \gets G - \{ u, v \}$
             \State $M_{KS} \gets M_{KS} \cup \{ \{u,v\} \}$ 
           \EndIf
      \EndWhile
      \State \textbf{return} $M_{KS}$
  \end{algorithmic}
\end{algorithm}

This fact is the idea behind Algorithm~\ref{ksalgo}, called the Karp-Sipser Algorithm ($KS$) \cite{KS}, which produces a matching $M_{KS}$. In every iteration of $KS$, among all vertices a vertex of \emph{minimum} degree is picked. 

{Regarding time complexity of the $KS$ algorithm, note the following connection with the Greedy one: in every iteration, $KS$ picks a vertex of minimum degree and by using a ``Heap" data structure, finding a vertex of minimum degree has complexity $O(1)$ \cite{cormen1990c}; therefore, the running time of $KS$ is linear. Further, since for real networks, the average degree is finite, for each iteration of the $KS$ algorithm (i.e. excluding up to $o(n)$ iterations) there are $O(n)$ many vertices of degree one, which implies that determining a degree one vertex takes on average in $O(1)$ steps.}

We can simplify the KS algorithm and search for a minimum degree vertex among vertices on only one side to derive Algorithm~\ref{oksalgo} that we call the One-sided Karp-Sipser ($OKS$) Algorithm and whose output we denote by $M_{OKS}$. {As we will see later in the analysis of real networks, the size of the matching given by $OKS$ is usually less than or equal to the size of the matching given by $KS$. The intuition behind this is as follows. It is possible to make a mistake in $OKS$ because of lack of degree one vertices on the right side, but if degree one vertices exist on the left side, $KS$ can still work optimally. Yet, later we will prove (asymptotic) optimality of both algorithms.}

\begin{algorithm}
  \caption{{\bf: One-Sided Karp-Sipser} \\ Input: $G=(L,R,E)$ \\ Output: matching $M_{OKS}(G)$}\label{oksalgo}
  \begin{algorithmic}
      \State $M_{OKS} \gets \emptyset$
      \While{$E \neq \emptyset$}
         \State let $v=\arg\!\min \limits_{w \in R} \deg(w)$
         \If{$\deg(v)=0$} 
            \State $G \gets G - \{ v \}$
          \ElsIf{for $u \in L$, $(u,v) \in E$}
             \State $G \gets G - \{ u, v \}$
             \State $M_{OKS} \gets M_{OKS} \cup \{ (u,v) \}$ 
           \EndIf
      \EndWhile
      \State \textbf{return} $M_{OKS}$
  \end{algorithmic}
\end{algorithm}

%% file: Numerical_Examples.tex
\section{{Numerical Examples}} \label{numerical}
{In this section, we present the results of selected numerical analyses on real world networks. We study the number of controllers, or equivalently the number of unmatched vertices for 10 different networks. This provides support to studying $DD$ random networks as a realistic model for control applications. In other words, using Fig. \ref{random_vs_real}, to find the number of controllers needed to structurally control a network, one can assume real networks are in fact $DD$ random networks.} 

{{The full description of these networks can be found in Table \ref{nwtable}.}
\begin{center}
\begin{table}
\centering
\begin{tabular}{ |m{1cm}| m{4cm}|m{4cm}|m{2cm}|m{2cm}|m{2cm}| } 
\hline
number & network & type & $|V|$  & $|E|$ & description \\
\hline
1 & leader2inter-st-college-student \cite{van2003evolution} & Social & 32 & 96  & directed \\
\hline
2 & celegansneural \cite{watts1998collective} & Neuronal & 297  & 2345 & directed \\
\hline
3 & p2p-Gnutella06 \cite{leskovec2007graph} & Internet & 8717 & 31525 & directed \\
\hline
4& polblogs \cite{adamic2005political}& WWW& 1490 & 19025  & directed \\
\hline
5 & prisoninter-st-prison-inmates \cite{van2003evolution} & Social & 67 & 182 & directed \\
\hline
6 & s208-st \cite{milo2002network} & Electronic Circuits & 122 & 189 & directed \\
\hline
7 & s420-st \cite{milo2002network} & Electronic Circuits & 252 & 399 & directed \\ 
\hline
8 & s838-st \cite{milo2002network} & Electronic Circuits & 512 & 819 & directed \\
\hline
9 & USpowergrid-n4941 \cite{watts1998collective} & Power Grid & 4941 & 13188 & undirected \\
\hline
10 & yeastinter-st \cite{milo2002network} & Transcriptional Regulatory & 688 & 1079  & directed \\
\hline
\end{tabular}
\caption{ {Real networks used in this paper} }
\label{nwtable}
\end{table}
\end{center}

\begin{table*}
\centering
\begin{tabular}{ | m{4cm} | m{.3cm}|m{.6cm}|m{1cm}|m{1cm}|m{.5cm}|m{.6cm}|m{.6cm}|m{1cm}|m{1cm}|m{1cm}| } 
\hline
~ & \multicolumn{10}{|c|}{Networks} \\
\hline Algorithms& 1  &        2  &     3   & 4  & 5  &  6  & 7   &  8  &   9  &   10 \\
\hline
$|V|-|M_{OKS}|$ & 6  &        51  &      5033   &      703  &        10  &        29  &        59    &     119   &      610   &      565 \\
\hline
$|V|-|M_{KS}|$ & 6  &  49  &      5033   &      702  &   9  &    29  &        59    &     119   &     577  &      565 \\
\hline
$|V|-|M^*|$ & 6  &  49  &      5033   &      702  &   9  &    29  &        59    &     119   &     575  &      565 \\
\hline 
average of $ |V|- |M_{OKS}| $ for random network &5.6&29.8&5006.8&616.5&10.1&24.7&48.6&100.5&454.9&557.7\\
\hline
$|V| =n $ & 32  & 297 & 8717 & 1490 & 67 & 122 & 252 & 512   & 4941  & 688 \\
\hline
\end{tabular}
\caption{{The number of controllers given by different algorithms for different networks as well as the average number of controllers for equivalent degree distribution random network and the size of the networks.} }
\label{resulttable}
\end{table*}

{Table \ref{resulttable} contains the results for 10 networks, including social, internet, web, electronic, neuronal, power grid and transcriptional regulatory networks, enumerated by the first row of the table. The second (third and fourth respectively) row is the number of controllers needed for structural controllability of the corresponding network if $OKS$ ($KS$ and Maximum Matching respectively) algorithm will be used. }

{The fifth row shows the average number of controllers for a random network generated by degree distribution model for random networks, i.e. input and output degrees of all vertices are the same as the original real-world network but the attachment of half-edges is random. 
This random attachment is based on a simple fact that every random permutation is superposition of large enough number of random swaps. 
Finally, the last row of the table shows the size of the networks. As seen in the table, $OKS$ and $KS$ perform very close to Maximum Matching for all networks. Moreover, network 9 is the only one for which the performance of $OKS$ is significantly different than $KS$. We suspect that this is because network 9 is the only undirected network.}

{The numerical results in Table \ref{resulttable} can be understood better by the following figures. Fig. \ref{KS_vs_MM} shows the performance of $OKS$ and $KS$ versus Maximum Matching, i.e. rows 2,3 of the Table \ref{resulttable} versus row 4. The similarity in performance between $OKS, KS$ and maximum matching algorithms is better depicted in Fig. \ref{KS_vs_MM}.}

{The number of controllers for equivalent degree distribution random network versus the original network, i.e. row 5 of the Table \ref{resulttable} versus row 4 is shown in Fig. \ref{random_vs_real}. According to this plot, degree distribution random networks are sufficiently realistic models for real-world networks in terms of the number of controllers needed for structural controllability. }
\begin{figure}[t!]
        \centering
        \scalebox{.65}
        {\includegraphics {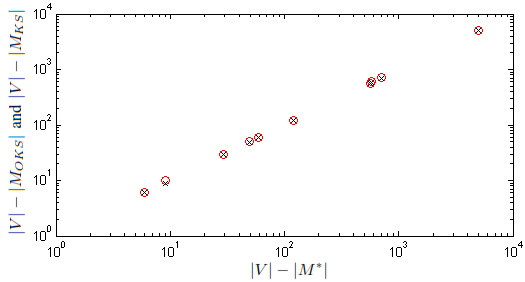}}
	\caption{{The number of controllers given by $OKS$ ( $\circ$ ) and $KS$ ( $\times$ ) versus the number of controllers given by Maximum Matching for 10 different real networks.} }
	\label{KS_vs_MM}
\end{figure}
\begin{figure}[t!]
        \centering
        \scalebox{.65}
        {\includegraphics {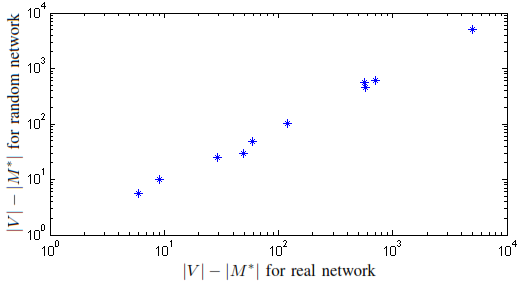}}
	\caption{{The number of controllers for randomly rewired networks versus the number of controllers for the original networks.} }
	\label{random_vs_real}
\end{figure}

%% file: matching_algorithms_in_random_networks.tex
\section{Matching Algorithms in Random Networks} \label{matching}

In this section, we present results about the asymptotic size of matchings produced by the algorithms presented above. {We first consider the general case where the asymptotic degree distribution of the random network is any arbitrary degree distribution with finite mean. Then, more detailed results regarding the special case for a Poisson asymptotic empirical degree distribution will be provided. Even though as seen before, the latter case is not common among real networks, because of classical interests on $ER$ and $UFS$ random networks and the $KS$ algorithm, as well as comprehensiveness, we also study it.} 

Before proceeding with the analysis of matching algorithms in an extensive class of random networks, we must ensure that the size of the matchings provided by either Greedy, KS, OKS algorithms or any \emph{maximum} matching algorithm has the Lipschitz property in order to have convergence of $\frac{|M_G(G)|}{n}, \frac{|M_{KS}(G)|}{n}, \frac{|M_{OKS}(G)|}{n}$ and $\frac{|M^*(G)|}{n}$ for random network $G$ where $M^*(G)$ is a maximum matching of network $G$. The following lemma establishes the desired Lipschitz property. The size of the matching provided by any of the above algorithms has the Lipschitz property due to the recursive nature of the algorithms. The Lipschitz property for the size of maximum matchings comes from their maximality regardless of the algorithm used to obtain the maximum matching (this follows easily from the definition of maximum matching).
\begin{lem}\label{LipProofHeuristics}
The real-valued functions $|M_G|, |M_{KS}|$, $|M_{OKS}|$ and $|M^*|$ which are the size of matchings provided by Greedy, KS, OKS and maximum matching algorithms, respectively, have the Lipschitz property.
\end{lem}

%%%%%%%%%%%%%%%%%%%%%%%%%%%%%%%%%%%%%%%%%%%%%%%%%%%%%%%%%%
\subsection{Arbitrary Degree Distribution} \label{optimality}
We establish the optimality of $OKS$ algorithm which immediately yields optimality of $KS$ as well, for reasons explained before. For this purpose, we follow in the footsteps of Karp and Sipser \cite{KS} and embed the dynamics of both input and output degree sequences as the algorithm proceeds in continuous time. This embedding provides differential equations governing the degree sequence vectors. However, in the general degree distribution case, unlike the classic Erdos-Renyi case, the differential equations are defined in arbitrarily high dimensions. So there is little hope of working in fixed small dimension as Karp and Sipser \cite{KS} did (their differential equations were 3 dimensional) and new ideas are needed. The key idea in our proof is to use the differential equations to show that the number of iterations when there is no degree one vertex (and so the algorithm can possibly make a mistake) is sublinear (in $n$) which means the fraction of unmatched vertices (or equivalently the relative size of the number of controllers) given by $OKS$ is asymptotically the same as that of maximum matching. Finally, a set of equations for the relative size of maximum matching according to asymptotic empirical input and output degree distributions is provided.
\begin{thm}[Asymptotic optimality of OKS algorithm]\label{OKS}
 For network $G=(L,R,E), |R| = |L| = n$ let $ | M_{OKS}(G) |$ and $| M^*(G)|$ be the size of matching given by OKS algorithm and the size of maximum matching respectively. Let $G$ be either $ER$, $UFS$ or $DD$ random network with finite average degree, i.e. $\lim \limits_{n \to \infty} \frac{ |E| }{n}=\sum \limits_{i=0}^ \infty i p_{in}(i)= \sum \limits_{i=0}^ \infty i p_{out}(i) < \infty $ (where $ p_{in},  p_{out}$ are asymptotic empirical degree distributions). Then 
$$\lim \limits_{n \to \infty} \frac{| M_{OKS}(G) |}{n}= \lim \limits_{n \to \infty} \frac{| M^*(G) |}{n}.$$
\end{thm}
{\begin{remark}
Under the assumptions of Theorem \ref{OKS}, one can show the asymptotic optimality of $KS$ as well:$\lim \limits_{n \to \infty} \frac{| M_{KS}(G) |}{n}= \lim \limits_{n \to \infty} \frac{| M^*(G) |}{n}$. We omit the proof here as it is very similar to the proof of Theorem \ref{OKS}.
\end{remark}}
Note that in Theorem \ref{conv} letting $\epsilon = n^{-r}$ for every $r>\frac{1}{2}$ the convergence holds. So the difference between $|M_{OKS}|$ and $|M^*|$ is $O(\sqrt n)$. Now the following questions arise: $(i)$ what is the size of maximum matching? $(ii)$ how can we compute the answer (asymptotically) without running the algorithm? The following theorem gives the size of maximum matching in terms of input and output degree distributions. For $u \in \left [ 0,1 \right )$ define moment generating functions:
$$\Phi_{in}(u)=\sum \limits_{i=0}^ \infty p_{in}(i)u^i, \Phi_{out}(u)=\sum \limits_{i=0}^ \infty p_{out}(i)u^i,$$
$$\phi_{in}(u)= \frac{1}{\mu} \Phi_{in}'(u)=\sum \limits_{i=1}^ \infty \frac{ip_{in}(i)}{\mu}u^{i-1},$$
$$\phi_{out}(u)= \frac{1}{\mu} \Phi_{out}'(u)=\sum \limits_{i=1}^ \infty \frac{ip_{out}(i)}{\mu}u^{i-1},$$
where $\mu = \sum \limits_{i=0}^ \infty i p_{in}(i)= \sum \limits_{i=0}^ \infty i p_{out}(i) $ is the average degree.
\begin{thm}[Size of Maximum Matching] \label{MMSize}
For (either $ER$, $UFS$ or $DD$) random network $G=(L,R,E), |R| = |L| = n$ if $\mu < \infty $ let $U^*$ be 
\begin{eqnarray}
U^* &=& \frac{1}{2} \Big [ \Phi_{in}(1-w_1) + \Phi_{in}(w_2)+ \Phi_{out}(1-w_3) -2 \notag \\ 
&+&\Phi_{out}(w_4)+\mu \left ( w_3 (1-w_2)+w_1(1-w_4) \right ) \Big ]  \label{optformula}
\end{eqnarray}
where $(w_1,w_2,w_3,w_4) \in \left [ 0, 1 \right ) ^4$ is the smallest solution of
$$ \phi_{out}(1-w_3)=1-w_2~~,~~  \phi_{in}(w_2)=w_3 ,$$
$$ \phi_{in}(1-w_1)=1-w_4 ~~,~~  \phi_{out}(w_4)=w_1 ,$$
then the asymptotic fraction of unmatched vertices is $U^*$: 
$$\lim \limits_{n \to \infty} 1-\frac{| M^*(G)|}{n}=U^*.$$ 
\end{thm}
Note that the formula \eqref{optformula} first appeared in Liu et al.'s work~\cite{Liu} and was
supported by numerical experiments. The result above formally proves its validity.

\subsection{Poisson Degree Distribution} \label{PoissonDD}
We first present selected asymptotic and non-asymptotic results about the fraction of matched vertices for the matching given by Greedy. Subsequently we generalize the results provided by Karp and Sipser \cite{KS} on the performance of $KS$.

Studying a simple algorithm (Greedy) on a simple random network $G$ (drawn from directed $ER$ where every edge $i \to j$ exists with probability $p$) allows us to obtain explicitly the \emph{non-asymptotic} probability mass function for $|M_G(G)|$. The following theorem also provides the asymptotic behavior of $|M_G(G)|$ for directed $ER$.
\begin{thm} [Greedy for directed $ER$]\label{greedy}
Let the network $G$ be directed $ER$ of size $n$. Then:
$$\mathbb{P}(|M_G(G)|=n-k)=\frac{\alpha_{n}(q)^2}{\alpha_{k}(q)^{2} \alpha_{n-k}(q)} q^{k^2} ,$$
where $q=1-p$ and $\alpha_{i}(q)=\prod\limits_{j=1}^i {(1-q^j)}$. For $ER(\lambda)$ (i.e. $np \to \lambda$) if $\lambda =0$ then $\lim \limits_{n \to \infty }\frac{|M_G(G)|}{n} = 0$. If $\lambda =\infty$ then $\lim \limits_{n \to \infty }\frac{|M_G(G)|}{n} = 1$. For $\lambda \in (0,\infty)$, $|M_G(G)|$ is asymptotically normal:
$$\mathcal N \left(n \frac{\lambda - \log (2 - e^{- \lambda})}{\lambda}, n \frac{1}{4 \lambda}\right).$$
\end{thm}
Some of the results in Theorem \ref{greedy} remain valid for a larger class of random networks including directed and undirected networks. 
\begin{thm} [Greedy for asymptotically Poisson degree distributions]\label{greedymean}
Assume $G$ is one of $ER(\lambda)$, $UFS(\lambda)$, $DD(p)$ and $DD(p,p)$ where probability distribution $p$ is $\Poisson(\lambda)$. If $\lambda =0$ (respectively $\lambda=\infty$) then $\lim \limits_{n \to \infty }\frac{|M_G(G)|}{n} = 0$ (respespectively $1$). For $\lambda \in (0,\infty)$, $\lim \limits_{n \to \infty }\frac{|M_G(G)|}{n} = 1 - \frac{ \log (2 - e^{- \lambda})}{\lambda}$.
\end{thm}
Similar results hold for $KS$. In fact, Karp and Sipser \cite{KS} proved that, for the classical undirected Erdos-Renyi random network (denoted by undirected $ER(\lambda), \lambda \in (0,\infty)$ here), $KS$ is optimal. They split the running of the algorithm into two phases. Phase 1 begins when the algorithms starts and finishes the first time there is no vertex of degree one in the network, when phase 2 starts and proceeds until the algorithm removes all edges from the network. For network $G$ let $U(G)$, $U_1(G)$ and $U_2$(G) be the number of vertices left unmatched {(became of degree zero before being removed from the network)} when running maximum matching, phase 1 and phase 2 respectively, $H=(V',E')$ be the remaining network at the beginning of phase 2. Hence $|M^*(G)|=n-U(G), |M_{KS}(G)|=n-U_1(G)-U_2(G), |M_{KS}(H)|=|V'|-U_2(G)$. Since there is no deviation from maximum matching as long as vertices of degree one exists, we have $U_1(G) \leq U(G) \leq U_1(G) + U_2(G)$. {Karp and Sipser} show $\frac{U_2(G)}{n} \to 0$ as $n \to \infty$ so the algorithm is optimal, i.e. $\lim \limits_{n \to \infty} \frac{|M_{KS}(G)|}{n}=\lim \limits_{n \to \infty} \frac{|M^*(G)|}{n}$. Further, they show $\frac{U_1(G)}{n} \to k(\lambda)$ and $\frac{|V'|}{n} \to h(\lambda)$ and find functions $k,h$ as $k(\lambda)= \frac{\gamma_* + \gamma^* + \gamma_* \gamma^*}{\lambda} -1 , h(\lambda)= \frac{(1-\gamma_*)(\gamma^* - \gamma_*)}{\lambda}$ where $\gamma_*$ is the smallest root of $\gamma=\lambda \exp (- \lambda e^{- \gamma})$ and $\gamma^*= \lambda e^ {- \gamma_*}$. For $\lambda \leq e$ we have $h(\lambda)=0$ because of $\gamma_* = \gamma^*$. In the following theorem, we generalize these results to a larger class of random networks. 
\begin{thm} [KS for asymptotically Poisson degree distributions]\label{KSgeneral}
Assume $G$ is one of $ER(\lambda)$, $UFS(\lambda)$, $DD(p)$ and $DD(p,p)$ where probability distribution $p$ is $\Poisson(\lambda)$. If $\lambda =0$ then $\lim \limits_{n \to \infty }\frac{|M_{KS}(G)|}{n} =\lim \limits_{n \to \infty }\frac{|M^*(G)|}{n} = 0$. If $\lambda =\infty$ then $\lim \limits_{n \to \infty }\frac{|M_{KS}(G)|}{n} =\lim \limits_{n \to \infty }\frac{|M^*(G)|}{n} = 1$. For $\lambda \in (0,\infty)$, we have
$$\lim \limits_{n \to \infty }\frac{|M_{KS}(G)|}{n} =\lim \limits_{n \to \infty }\frac{|M^*(G)|}{n} = 1 - k( \lambda) .$$
Furthermore, $\frac{U_2(G)}{n} \to 0$ , $\frac{U_1(G)}{n} \to k(\lambda)$ and $\frac{|V'|}{n} \to h(\lambda)$ as $n \to \infty$.
\end{thm}
{\begin{remark}
Note that the results in Theorem \ref{KSgeneral} are consistent with Theorem \ref{MMSize} as follows. Letting $p_{in},p_{out}$ be $\Poisson(\lambda)$, calculations show $U^*=k(\lambda)$ because of $\mu=\lambda$, $\Phi_{in}(u)= \Phi_{out}(u)=\phi_{in}(u)= \phi_{out}(u)= \exp \left ( \lambda (u-1) \right )$.
\end{remark}}
Results presented in Theorems \ref{greedymean} and \ref{KSgeneral} are based on the fact that in all mentioned random networks, the asymptotic empirical degree distribution is \emph{Poisson}.

%% file: conclusion.tex
\section{{Conclusion}}

{We proved that the $OKS$ algorithm is (asymptotically) optimal for determining a set of vertices where controllers should be applied to. Indeed, first benefiting from the connection between structural controllability of networks and maximum matching problems (minimum inputs theorem) we introduced simple fast matching algorithms $OKS$ and $KS$. Further, using topologies extracted from real networks, we empirically showed that the minimum number of controllers for structural controllability heavily depends on the degree distribution of the network, which in turn implies that the assumption of a random network with specified degree distribution is reasonable for many real world networks.
Finally, new proof techniques introduced in this study enable the rigorous analysis of the the performance of a class of fast matching algorithms for random networks.}

{Ruths and Ruths \cite{ruths2014control} showed that existing random network models, while capturing the key features for predicting the minimum number of controllers, are not predicting more detailed control profiles of real networks. This calls for the development of new random network models that match control profiles, and associated fast control algorithms.}

%% file: appendix.tex
\appendices

\section{Proof of Lemma \ref{RNDegDist}}
For an arbitrary $\epsilon>0$, let $N$ be large enough such that $\sum \limits_{k=N+1}^\infty kp(k) < \frac{\epsilon}{2}$. Further, for  random network $G=(V,E), |V|=n$, remove some edges from $G$ in order to have no vertex of degree larger than $N$ to get network $G'=(V,E')$. Now for every vertex $i \in V$ there are possibly two reasons for the difference between $D_i$ and $\deg_{G'}(i)$:
\begin{itemize}
\item
the set of edges we removed from $G$ to get $G'$

\item
multiple edges in the network $G'$
\end{itemize}
so,
\begin{eqnarray} \label{totaldiff}
\Big | |\{ i \in V : D_i=k\}| - |\{ i \in V : \deg_{G'}(i)=k\}| \Big | \leq |E|-|E'| + \sum \limits_{i=1}^n 1_{M_i \neq 0}
\end{eqnarray}
where $M_i$ is the number of multiple edges in $G'$ connected to vertex $i$. There are at most $N \choose 2$ pairs of half-edges connected to vertex $i$ and for every two of them the probability of the outcome $A_j$ that they both are connected to vertex $j$ is at most $\frac{{N \choose 2}}{{D \choose 2}}$ where $D=2|E'|$. So,
$$ M_i \leq {N \choose 2} \sum \limits_{j=1}^n 1_{A_j} \quad,\quad \mathbb P (A_j | D) \leq \frac{{N \choose 2}}{{D \choose 2}}.$$
By Markov's inequality for any $\delta > 0$ we have: 
\begin{eqnarray*}
\mathbb P \left ( \frac{1}{n}\sum \limits_{i=1}^n 1_{M_i \neq 0} > \delta \Big | D \right ) 
\leq \frac{1}{n \delta}\sum \limits_{i=1}^n \mathbb P (M_i \neq 0 \big | D ) = \frac{1}{n \delta}\sum \limits_{i=1}^n \mathbb P (M_i \geq 1 \big | D ) 
\leq \frac{1}{n \delta}\sum \limits_{i=1}^n \mathbb E (M_i \big | D ) \leq \frac{2n{N \choose 2}^2}{\delta (D-1)^2} .
\end{eqnarray*}
But $\lim \limits_{n \to \infty} \frac{|E|}{n} = \mu$ and 
\begin{eqnarray} \label{edgediff}
\lim \limits_{n \to \infty} \frac{|E|-|E'|}{n} < \frac{\epsilon}{2}
\end{eqnarray}
(since $\sum \limits_{k=N+1}^\infty kp(k) < \frac{\epsilon}{2}$ ) imply $\lim \limits_{n \to \infty} \frac{D-1}{n} > \mu - \epsilon$ i.e. 
\begin{eqnarray} \label{multiedge}
\frac{1}{n}\sum \limits_{i=1}^n 1_{M_i \neq 0} \to_P 0
\end{eqnarray}
On the other hand, for all $k=1,2, \ldots$ by the Law of Large Numbers:
\begin{eqnarray}\label{LLN}
\lim \limits_{n \to \infty} \frac{| \{ i \in V : D_i=k \} |}{n}=p(k)
\end{eqnarray}
Finally, because the only reason for the difference between $\deg_{G}(i)$ and $\deg_{G'}(i)$ is the set of edges we removed from $G$ to get $G'$ we have:
\begin{eqnarray}\label{degdiff}
\Big | |\{ i \in V : \deg_{G}(i)=k\}| - |\{ i \in V : \deg_{G'}(i)=k\}| \Big | \leq |E|-|E'|
\end{eqnarray}
Putting \eqref{totaldiff}, \eqref{edgediff}, \eqref{multiedge}, \eqref{LLN}, \eqref{degdiff} all together: 
$$\lim \limits_{n \to \infty} \mathbb P \left ( \Big | \frac{|\{ i \in V : \deg_{G}(i)=k\}|}{n} - p(k) \Big | > \epsilon \right ) = 0$$
Since the function $f(G)=|\{ i \in V : \deg_{G}(i)=k\}|$ has the Lipschitz property by Theorem \ref{conv}, a.s. convergence holds as well i.e. $\hat p(k)= p(k)$

For $DD(p_{in},p_{out})$, the changes in the asymptotic empirical degree distributions due to omission of $$\max \{ \sum \limits_{i=0}^n D^{(in)}_i, \sum \limits_{i=0}^n D^{(out)}_i\} - \min \{ \sum \limits_{i=0}^n D^{(in)}_i, \sum \limits_{i=0}^n D^{(out)}_i \}$$ additional half-edges is at most $\frac{1}{n} |\sum \limits_{i=0}^n D^{(in)}_i - \sum \limits_{i=0}^n D^{(out)}_i| \to 0$ because by the Law of Large Numbers $\lim \limits_{n \to \infty} \frac{1}{n} \sum \limits_{i=0}^n D^{(in)}_i =\lim \limits_{n \to \infty} \frac{1}{n} \sum \limits_{i=0}^n D^{(out)}_i = \mu$ as $n \to \infty$. The contribution of multiple edges is asymptotically zero similarly.

\section{Proof of Theorem \ref{conv}}
To prove convergence and concentration inequalities for real-valued functions for random networks we use some classical notions of probability such as martingale difference sequences. Rhee \cite[Theorem 1]{wansoo} presents a concentration inequality for martingale difference sequences. Here we use a slightly more general version of it. The proof presented by Rhee \cite{wansoo} is valid by following the same line of reasoning.
\begin{thm}\label{mdsineq}
Let $X_i , i=1, 2, \ldots, k$ be a martingale difference sequence. If $ \max \limits_{1 \leq i \leq k} \| X_i \|_{\infty} \leq M < \infty$ and $ \sum \limits_{i=1}^k \mathbb{E}(X_i^2 | \mathcal{F}_{i-1}) \leq a^2< \infty $ then for all $t \geq 0$
$$ \mathbb{P}( | \sum \limits_{i=1}^k X_i | > t ) \leq 2 \exp (- \frac {a^2}{M^2} \rho (\frac {Mt}{a^2} )) $$
where $\rho(x)=(1+x) \log(1+x)-x$ for $x \geq 0$.
\end{thm}
We prove the following stronger theorems:
\begin{thm} \label{convGER}
Assume $f$ has the Lipschitz property and $G=(V,E)$ is a random network in which $V= \{1,2, \ldots, n \}$ and for $i,j \in V$, the edge $i \to j$ exists independently with probability $p_{ij}^{(n)}$. If 
\begin{eqnarray}
\label{ger1}
\limsup \limits_{n \rightarrow \infty} \frac{1}{n^2} \sum \limits_{i,j=1}^n p_{ij}^{(n)} (1-p_{ij}^{(n)}) = 0
\end{eqnarray}
then $\frac{ f(G) - \mathbb{E}(f(G))}{n} \rightarrow _P 0$ as $n \rightarrow \infty$. Furthermore, if
\begin{eqnarray}
\label{ger2}
\limsup \limits_{n \rightarrow \infty} \frac{\log n}{n^2} \sum \limits_{i,j=1}^n p_{ij}^{(n)} (1-p_{ij}^{(n)}) =0
\end{eqnarray}
then $\frac{f(G) - \mathbb{E}(f(G))}{n} \rightarrow 0 \as$ as $n \rightarrow \infty$. When the average degree is finite (e.g. $ER(\lambda)$ for $\lambda < \infty$) the rate of convergence is exponential. In general given \begin{eqnarray}
\label{ger3}
\sup \limits_{n \geq 1} \frac{1}{n} \sum \limits_{i,j=1}^n p_{ij}^{(n)} (1-p_{ij}^{(n)}) \leq C < \infty
\end{eqnarray}
we have:
$$ \mathbb{P}(  \frac{| f(G) - \mathbb{E}(f(G)) |}{n} > \epsilon ) \leq 2 \exp (- nC \rho (\frac { \epsilon}{C} ) )$$
\end{thm}
\begin{proof}
To have more convenient notation we enumerate all possible edges $i \to j$ from $1$ to $k$ where $k=n^2$ for directed $ER$ and $k=\frac{n(n+1)}{2}$ for undirected $ER$. Indeed, $Z_i, i=1, 2 , \ldots, k$ are independent Bernoulli random variables showing the existence of edges, i.e. if $Z_i=1$ the corresponding edge exists and if $Z_i=0$ the corresponding edge does not exist. Let $p_i, i=1,2, \ldots$ be edge existence probabilities, $\mathbb{P}(Z_i=1)$, $\mathcal{F}_0$ be the trivial sigma-field and for $i=1, \ldots, k$ let $\mathcal{F}_i = \sigma (Z_1, \ldots, Z_i)$. Now define a martingale difference sequence as $X_i= \mathbb{E}(f(G) | \mathcal{F}_i) - \mathbb{E}(f(G) | \mathcal{F}_{i-1}) , i=1, \ldots, k$ so $\mathbb{E}(f(G) | \mathcal{F}_k)=f(G) , \mathbb{E}(f(G) | \mathcal{F}_0)=\mathbb{E}(f(G))$. Define:
$$ U_i=\mathbb{E}(f(G) | Z_1, \ldots, Z_{i-1}, Z_i=1) $$
$$ V_i=\mathbb{E}(f(G) | Z_1, \ldots, Z_{i-1}, Z_i=0) $$
Thus,
$$\mathbb{E}(f(G) | Z_1, \ldots, Z_{i-1})=p_i U_i + (1-p_i) V_i $$
$$X_i= \left\{
\begin{array}{ll}
(1-p_i)(U_i-V_i)  \;\;\;\;\;\; Z_i=1 \\
-p_i(U_i-V_i)  \;\;\;\;\;\;\;\;\;\;\;\; Z_i=0
\end{array}
\right.$$
$$ \mathbb{E}(X_i^2 | \mathcal{F}_{i-1}) = p_i (1-p_i) (U_i-V_i)^2 $$
Since $f$ has the Lipschitz property, $| U_i - V_i | \leq M$ ($M=1$ for directed $ER$ and $M=2$ for undirected $ER$):
$$ \max_{1 \leq i \leq k} \| X_i \|_{\infty} \leq \max_{1 \leq i \leq k} \max \{ p_i , 1-p_i\} \leq 1 $$
$$ \sum \limits_{i=1}^k \mathbb{E}(X_i^2 | \mathcal{F}_{i-1}) \leq \sum \limits_{i=1}^k p_i (1-p_i) \leq a^2< \infty $$
By Theorem \ref{mdsineq} for all $\epsilon > 0$ we have:
\begin{eqnarray}
\label{greproofineq}
\mathbb{P}(  \frac{| f(G) - \mathbb{E}(f(G)) |}{n} > \epsilon ) \leq 2 \exp (- a^2 \rho (\frac {n \epsilon}{a^2} ))
\end{eqnarray}
$ \lim \limits_{x \rightarrow 0} \frac{\rho (x)}{x^2}=\frac{1}{2}$ implies:  
$$\mathbb{P}(  \frac{| f(G) - \mathbb{E}(f(G)) |}{n} > \epsilon ) \leq 2 \exp (- \epsilon ^2 \frac {n^2}{2a^2})$$ 
now if \eqref{ger1} holds then $\frac{ f(G) - \mathbb{E}(f(G))}{n} \rightarrow _P 0$ as $n \rightarrow \infty$. To show the a.s. convergence, note that \eqref{ger2} gives$\sum \limits_{n=1}^{\infty} \mathbb{P}(  \frac{| f(G) - \mathbb{E}(f(G)) |}{n} > \epsilon ) < \infty$ so by the Borel-Cantelli lemma $\frac{f(G) - \mathbb{E}(f(G))}{n} \rightarrow 0 \as $ as $n \rightarrow \infty$.

Finally using \eqref{ger3} we can let $a^2=nC$ in \eqref{greproofineq} to obtain:
$$ \mathbb{P}(  \frac{| f(G) - \mathbb{E}(f(G)) |}{n} > \epsilon ) \leq 2 \exp (- nC \rho (\frac { \epsilon}{C} ) )$$
\end{proof}

Theorem \ref{convGER} provides convergence and concentration inequality for $ER$ random networks. A corollary of Theorem \ref{convGER} is that the obtained results are valid for undirected random networks of Chung-Lu type (see page 82 of Durrett \cite{DurretBook}). In a Chung-Lu random network the edge between $i$ and $j$ exists with probability $\frac{w_i w_j}{\sum \limits_{k=1}^n w_k}$ for the set of weights $w_1, \ldots, w_n$. Conditions \eqref{ger1}, \eqref{ger2}, \eqref{ger3} for a Chung-Lu random network will be $\limsup \limits_{n \to \infty} \frac{\overline{w}_n}{n} =0$, $\limsup \limits_{n \rightarrow \infty} \frac{\overline{w}_n \log n}{n} =0$ and $\sup \limits_{n \geq 1} \overline{w}_n < \infty $ respectively, where $\overline{w}_n$ is the average weight $\overline {w}_n=\frac{1}{n} \sum \limits_{k=1}^n w_k$.
\begin{thm}\label{convufs} 
Assume $f$ has the Lipschitz properrty and $G=(V,E), |E|=k_n$ is $UFS$. If 
\begin{eqnarray}
\label{ufs1}
\limsup \limits_{n \rightarrow \infty} \frac{k_n}{n^2} = 0
\end{eqnarray}
then $\frac{ f(G) - \mathbb{E}(f(G))}{n} \rightarrow _P 0$ as $n \rightarrow \infty$. If furthermore,
\begin{eqnarray}
\label{ufs2}
\limsup \limits_{n \rightarrow \infty} \frac{k_n \log n}{n^2}=0
\end{eqnarray}
then $\frac{f(G) - \mathbb{E}(f(G))}{n} \rightarrow 0 \as $ as $n \rightarrow \infty$. The rate of convergence is exponential when the average degree is finite (e.g. $UFS(\lambda)$ for $\lambda < \infty$). Namely $\sup \limits_{n \geq 1} \frac{k_n}{n} \leq C < \infty$ implies:
\begin{eqnarray}
\label{ufs3}
\mathbb{P}(  \frac{| f(G) - \mathbb{E}(f(G)) |}{n} > \epsilon ) \leq 2 \exp (-n \frac {\epsilon ^2}{C})
\end{eqnarray}
\end{thm}

\begin{proof} Rhee \cite[Theorem 4]{wansoo} shows 
$$\mathbb{P}(  | f(G) - \mathbb{E}(f(G)) | > t ) \leq 2 \exp (-\frac {t^2}{k_n}) $$
If \eqref{ufs1} holds then letting $t=n \epsilon$ we have $\exp (-\frac {t^2}{k_n}) \to 0$. If \eqref{ufs2} holds then $\sum \limits_{n=1}^{\infty} \mathbb{P}(  \frac{| f(G) - \mathbb{E}(f(G)) |}{n} > \epsilon ) < \infty$ so by the Borel-Cantelli Lemma $\frac{f(G) - \mathbb{E}(f(G))}{n} \rightarrow 0 \as $ as $n \rightarrow \infty$. To show \eqref{ufs3} it suffices to let $t=n \epsilon$.
\end{proof}
\begin{thm}\label{convDD}
Let $G=(V,E)$ be $DD(p_{in},p_{out})$ or $DD(p)$ (in the recent case $p_{in}=p_{out}=p$). Assuming real-valued function $f$ has the Lipschitz property if 
\begin{eqnarray}
\label{DD1}
\limsup \limits_{n \rightarrow \infty} \frac{| E |}{n^2}  = 0
\end{eqnarray}
then $\frac{ f(G) - \mathbb{E}(f(G))}{n} \rightarrow _P 0$ as $n \to \infty$. If in addition
\begin{eqnarray}
\label{DD2}
\limsup \limits_{n \rightarrow \infty} \frac{| E | \log n}{n^2}  = 0
\end{eqnarray}
then $\frac{f(G) - \mathbb{E}(f(G))}{n} \rightarrow 0 \as. $ as $n \to \infty$. When the average degree is finite i.e.  
\begin{eqnarray}
\label{DD3}
\sup \limits_{n \geq 1} \frac{| E | }{n}  \leq C < \infty
\end{eqnarray}
the rate of convergence is exponential:
$$ \mathbb{P}(  \frac{| f(G) - \mathbb{E}(f(G)) |}{n} > \epsilon ) \leq 2 \exp (- n C \rho (\frac { \epsilon}{4C} ) )$$
\end{thm}
\begin{proof}
To form an inequality like \eqref{greproofineq} let $k$ be the number of half-edges before being paired (so $k$ is the number of directed edges - every undirected edge is two directed edges). Enumerate half-edges and for $i=1, 2 , \ldots, k$ let $Z_i$  be the random variable indicating the vertex whose half-edge the half-edge $i$ is paired to i.e. $Z_i=j , j \in \{1, \ldots, n \}$ if half-edge $i$ is paired to a half-edge of vertex $j$. Let $\mathcal{F}_0$ be trivial sigma-field and for $i=1, \ldots, k$ let $\mathcal{F}_i = \sigma (Z_1, \ldots, Z_i)$. Now define a martingale difference sequence as $X_i= \mathbb{E}(f(G) | \mathcal{F}_i) - \mathbb{E}(f(G) | \mathcal{F}_{i-1}) , i=1, \ldots, k$ and $U_{ij}, P_{ij}$ as :
$$ U_{ij}=\mathbb{E}(f(G) | Z_1, \ldots, Z_{i-1}, Z_i=j) $$
$$ P_{ij}=\mathbb{P}(Z_i=j | Z_1, \ldots, Z_{i-1}) $$
Hence $U_i = \mathbb{E}(f(G) | Z_1, \ldots, Z_{i-1})= \sum \limits_{j=1}^n P_{ij} U_{ij}$ and $X_i= U_{ij}-U_i$ whenever $Z_i=j$. On the other hand,
$$ \mathbb{E}(X_i^2 | \mathcal{F}_{i-1}) = \sum \limits_{j=1}^n P_{ij} (U_{ij}-U_i)^2 $$
Lipschitz property of $f$ implies $| U_{ij} - U_{il} | \leq M $ ($M \leq 4$ for undirected case and $M \leq 2$ for directed case). Hence for all $i=1, \ldots, k$:
\begin{eqnarray*} \| X_i \|_{\infty} \leq \max_{1 \leq j \leq n} \| U_{ij}-U_i \|_{\infty} 
\leq \max_{1 \leq j \leq n} \sum \limits_{l=1}^n P_{il} \| U_{ij}-U_{il} \|_{\infty} \leq M < \infty \\
\sum \limits_{i=1}^k \mathbb{E}(X_i^2 | \mathcal{F}_{i-1}) \leq \sum \limits_{i=1}^k \sum \limits_{j=1}^n P_{ij} M^2 
\leq \sum \limits_{i=1}^k M^2 \leq kM^2 \leq a^2< \infty
\end{eqnarray*}
By Theorem \ref{mdsineq} we have:
$$ \forall \epsilon \geq 0 \: \: \: \mathbb{P}(  \frac{| f(G) - \mathbb{E}(f(G)) |}{n} > \epsilon ) \leq 2 \exp (- \frac {a^2}{M^2} \rho (\frac {Mn \epsilon}{a^2} )) $$
Now \eqref{DD1}, \eqref{DD2}, \eqref{DD3} give $\limsup \limits_{n \to \infty} \frac{a^2}{n^2}  = 0$, $\limsup \limits_{n \to \infty} \frac{a^2 \log n}{n^2}  = 0$ and $\sup \limits_{n \geq 1} \frac{a^2 }{n}  \leq 16 C < \infty$ respectively.
\end{proof}

\section{Proof of Theorem \ref{greedy}}

Letting $M^{(n)}=n-|M_G(G)|$ be the number of unmatched vertices $M^{(n)}=\sum \limits_{i=1}^n M_i$ where 
$$M_i= \left\{
\begin{array}{ll}
1  \;\;\;\; \text{    the vertex picked in $i$-th iteration is unmatched} \\
0  \;\;\;\; \text{    the vertex picked in $i$-th iteration is matched}
\end{array}
\right.$$
Now note that according to the algorithm $$\mathbb{P}(M_i=1|M_{1}, \ldots, M_{i-1})=q^{n-i+1+ \sum \limits_{j=1}^{i-1} M_j}$$
For example $\mathbb{P}(M^{(n)}=0)=\mathbb{P}(M_1=0, \ldots, M_n=0)=\prod \limits_{i=1}^n (1-q^{n-i+1})= \alpha_{n}(q)$. We have 
$$\mathbb{P}(M^{(n)}=k)= \sum \limits_{|I|=k} \mathbb{P}(M_i=1_{\{i \in I \}}, \forall \: i=1, \ldots, n).$$
If $I=\{ i_1, \ldots, i_k\}$ then
\begin{eqnarray*}
\mathbb{P} \left ( M_i=1_{\{i \in I \}}, \forall  i=1, \ldots, n \right ) 
&=& (1-q^n)\ldots(1-q^{n-i_1+1})q^{n-i_1+1} 
(1-q^{n-i_1})\ldots(1-q^{n-i_k+1})q^{n-i_k+1}\ldots(1-q^{k+1}) \\
&=& \frac{\alpha_n(q)}{\alpha_k(q)} \prod \limits_{j=1}^k q^{n-i_j+1}
\end{eqnarray*}
Thus:
$$\mathbb{P}(M^{(n)}=k)=
\frac{\alpha_n(q)}{\alpha_k(q)} q^{\frac{1}{2}k(k-1)}\sum \limits_{1 \leq i_1 < \ldots<i_k \leq n } q^{\sum \limits_{j=1}^k {i_j}} $$
But:
\begin{eqnarray*}
\sum \limits_{1 \leq i_1 < \ldots<i_k \leq n } q^{\sum \limits_{j=1}^k {i_j}} 
&=& \sum \limits_{i_k=k}^n q^{i_k} \sum \limits_{i_{k-1}=k-1}^{i_k-1} q^{i_{k-1}}  \ldots \sum \limits_{i_1=1}^{i_2-1} q^{i_1}  \\
&=& \sum \limits_{i_k=k}^n q^{i_k}  \ldots \sum \limits_{i_2=2}^{i_3-1} q^{i_2} (q \frac{1-q^{i_2-1}}{1-q}) \\
&=& \sum \limits_{i_k=k}^n q^{i_k}  \ldots \sum \limits_{i_2=2}^{i_3-1} q^{i_2} \frac{q^1}{\alpha_1(q)} \frac{\alpha_{i_2-1}(q)}{\alpha_{i_2-2}(q)} \\
&=& \sum \limits_{i_k=k}^n q^{i_k}  \ldots \sum \limits_{i_3=3}^{i_4-1} q^{i_3}  \frac{q^{1+2}}{\alpha_2(q)} \Big [ (1+q)(1-q^{i_3-2}) -q(1-q^{2i_3-5}) \Big ] \\
&=& \sum \limits_{i_k=k}^n q^{i_k}  \ldots \sum \limits_{i_3=3}^{i_4-1} q^{i_3} \frac{q^{1+2}}{\alpha_2(q)} \frac{\alpha_{i_3-1}(q)}{\alpha_{i_3-3}(q)} \\
&=&  \ldots = \frac{\alpha_n(q)}{\alpha_k(q) \alpha_{n-k}(q)} q^{\frac{1}{2}k(k+1)}
\end{eqnarray*}
which establishes the result. Now to show $\lim \limits_{n \to \infty} \frac{M^{(n)}}{n}= \frac{ \log (2 - e^{- \lambda})}{\lambda}$ for $ER(\lambda)$, for large $n$ let $(1-\frac{\lambda}{n})^j= \exp(-\frac{\lambda}{n}j)$ to have:
\begin{eqnarray*}
\frac{\mathbb{P}(M^{(n)}=k+1)}{\mathbb{P}(M^{(n)}=k)}= \frac{1-q^{n-k}}{(1-q^{k+1})^2} q^{2k+1}= 
\frac{1- \exp \left (- \frac{\lambda}{n}(n-k) \right )}{ \left [ 1- \exp \left ( - \frac{\lambda}{n}(k+1) \right ) \right ]^2} \exp \left ( - \frac{\lambda}{n}(2k+1) \right).
\end{eqnarray*}
Denote $\beta_k= \exp(- \frac{\lambda}{n}k)$ to obtain:
$$\frac{\mathbb{P}(M^{(n)}=k+1)}{\mathbb{P}(M^{(n)}=k)}= \frac{\beta_k^2 (1-e^{-\lambda} \beta_k ^{-1}) }{e ^ {\frac{\lambda}{n}} - 2 \beta_k + e^{- \frac{\lambda}{n}} \beta_k^2}.$$
Therefore $\frac{\mathbb{P}(M^{(n)}=k+1)}{\mathbb{P}(M^{(n)}=k)} > 1 $ if and only if $ (e^{- \frac{\lambda}{n}}-1) \beta_k^2  + (e^{-\lambda}-2) \beta_k +   e ^ {\frac{\lambda}{n}} < 0$ i.e. the probability mass function is unimodal and as $n \rightarrow \infty$, $\mathbb{P}(M^{(n)}=k+1) > \mathbb{P}(M^{(n)}=k)$ if and only if $\beta_k^{-1} < 2-e ^ {-\lambda}$ which is equivalent to $ k < \frac{ \log (2 - e ^{-\lambda})}{\lambda} n$ so letting $k^* = \arg \max \limits_{0 \leq k \leq n} \mathbb{P}(M^{(n)}=k)$ we get $\mathbb P (|M^{(n)}-k^*| > n \epsilon) \to 0$ for all $\epsilon > 0$ as $n \to \infty$ i.e. $\frac{M^{(n)}}{n} \to _P \frac{ \log (2 - e ^{-\lambda})}{\lambda}$. By Theorem \ref{convGER}, a.s. convergence holds as well. To prove the asymptotic normality we will show
$$ \lim \limits_{n \to \infty} \log \frac{\mathbb P (M^{(n)}=k(t))}{\mathbb P (M^{(n)}=k(0))} =  - \frac{t^2}{2 \sigma ^2}$$
where $t = \lim \limits_{n \to \infty} \frac{k(t)-n\mu}{\sqrt{n}}, k(0)=k^* \approx n \mu, \mu=\frac{ \log (2 - e ^{-\lambda})}{\lambda}, \sigma ^2 = \frac{1}{4 \lambda}$. We have:
\begin{eqnarray} \label{desired}
\log \mathbb P (M^{(n)}=k(t)) = \sum \limits_{j=1}^{n-k(t)} f_j(\frac{t}{\sqrt{n}}) - \lambda n \mu ^2 - \lambda t^2 - 2 \lambda \sqrt{n} \mu t
\end{eqnarray}
where the functions $f_j: \mathbb R \to \mathbb R, j=1,2, \ldots$ are $f_j(x)=2 \log (1- e^{-\lambda (\mu+ j/n) - \lambda x}) - \log (1 - e^{- \lambda j/n})$. Note that
\begin{eqnarray} \label{mu}
e^{- \lambda \mu }= (2- e^{-\lambda} )^{-1}.
\end{eqnarray}
But for $a_j = e ^{- \lambda (\mu + j/n)} <1$ and some $ B < \infty$  
\begin{eqnarray*}
f_j(0) &=& 2 \log (1- e^{-\lambda (\mu+ j/n)}) - \log (1 - e^{- \lambda j/n}) \\
f_j '(0) &=&  \frac{2 \lambda a_j}{1-a_j} \\
f_j ''(0) &=&  \frac{-2 \lambda ^2 a_j}{(1-a_j)^2} \\
|f_j '''(x)| &<&  B .
\end{eqnarray*}
Now writing $f_j(\frac{t}{\sqrt{n}})= f_j(0) + f_j '(0) \frac{t}{\sqrt{n}} + f_j ''(0) \frac{t^2}{2n} + f_j '''(\frac{t^*}{\sqrt{n}}) \frac{t^3}{6n\sqrt{n}}$ for some $t^* \in [0,t]$ and using \eqref{mu}, because $\frac{1}{n} (n - k(t) ) \to 1- \mu$ definition of Riemann integral implies:
\begin{eqnarray*}
\frac{1}{n} \sum \limits_{j=1}^{n-k(t)} f_j '(0) &\to & 2\lambda \int \limits_{0}^{1- \mu} \frac{1}{e^{\lambda \mu + \lambda x}-1} dx = 2 \lambda \mu \\
\frac{1}{n} \sum \limits_{j=1}^{n-k(t)} f_j ''(0)
& \to & -2 \lambda ^2 \int \limits_{0}^{1- \mu} \frac{e^{-\lambda (\mu + x)}}{(1-e^{- \lambda (\mu +x)})^2} dx = -2 \lambda \\
\frac{1}{n\sqrt{n}} \sum \limits_{j=1}^{n-k(t)} f_j '''(\frac{t^*}{\sqrt{n}}) & < & \frac{B}{\sqrt{n}} \to 0 .
\end{eqnarray*}
Note that $| \frac{1}{n} \sum \limits_{j=1}^{n-k(t)} f_j '(0) - 2 \lambda \mu | \leq \frac{(1 - \mu)^2}{n} \sup \limits_{\mu \leq x \leq 1} \frac{\partial }{\partial x} \frac{1}{e^{\lambda x}-1}$ and $\sup \limits_{\mu \leq x \leq 1} \frac{\partial }{\partial x} \frac{1}{e^{\lambda x}-1}= \sup \limits_{\mu \leq x \leq 1} \frac{\lambda e^{\lambda x}}{(e^{\lambda x}-1)^2} < \infty$ imply $\frac{1}{\sqrt{n}} \sum \limits_{j=1}^{n-k(t)} f_j '(0) - 2 \lambda \mu \sqrt{n} \to 0$. Now plugging all in \eqref{desired} we get the desired result since $\log \mathbb P (X=k(0)) = \sum \limits_{j=1}^{n-k(0)} f_j(0) - \lambda n \mu ^2$.

\section{Proof of Lemma \ref{LipProofHeuristics}}
To prove the Lipschitz property for maximum matching let $G=(L,R,E), G'=(L,R,E') , E'=E \cup \{ e \}, |L|=|R|=n$ and let $M \subset E'$ be a maximum matching in $G'$. If $e \notin M$ then $M$ is a maximum matching in $G$. If $e \in M$ then $M-\{e \}$ is a matching in $G$ so the size of maximum matching is at least $|M-\{e \}|$. Thus $|M^*(G')|-1 \leq |M^*(G)| \leq |M^*(G')|$.

For the Greedy algorithm, if $n=1$ then clearly $|M_G(G')|, |M_G(G)|$ are both either 0 or 1 so $\Big | |M_G(G')| - |M_G(G)| \Big | \leq 1$. Assume function $|M_G|$ has the Lipschitz property for all networks of vertex size $n-1$ and define networks $H,H'$ as the networks $G,G'$ after the first iteration: $H=G - \{ u,v \}, H'=G' - \{ u',v' \}; u,u' \in L ; v,v' \in R$ so both $H,H'$ have $n-1$ vertices. 

If $u=u'$ and $v=v'$ then networks $H,H'$ are exactly the same except potentially the new edge $e$ which can be added to $H$ in order to get $H'$. Since $H,H'$ both have $n-1$ vertices, by the assumption $|M_G(H') - M_G(H)| \leq 1$. On the other hand, $M_G(G) = M_G(H) \cup \{ (u,v) \}, M_G(G') = M_G(H') \cup \{ (u,v) \}$ imply $|M_G(G)| = |M_G(H)|+1, |M_G(G')| = |M_G(H')|+1$ i.e. $|M_G(G') - M_G(G)| \leq 1$.

If $v \neq v'$ then the new edge $e$ is connected to $v' \in R , u' \in L$ and $deg_G (v')=0, deg_{G'}(v')=1$. Now there are two possible cases:

\begin{enumerate}
\item
vertex $u' \in L$ is not used in $M_G(G)$ i.e. there is not any vertex $w \in R$ in $M_G(G)$ such that $u'$ is matched to $w$ by Greedy algorithm applied to $G$. Therefore iterations of the Greedy algorithm after the first iteration will not change once the new edge $e$ is added. In this case $M_G(G)= M_G(H')$ and $ M_G(G')= M_G(H') \cup \{ (u',v') \}$.

\item
vertex $u' \in L$ is used in $M_G(G)$ i.e. there is a vertex $w \in R$ such that $u'$ is matched to $w$ by Greedy algorithm applied to $G$: $(u',w) \in M_G(G)$ . Therefore the only iteration of the Greedy algorithm which will change once the new edge $e$ is added is the one that vertex $w \in R$ is picked. In this case $M_G(G)= M_G(H') \cup \{ (u',w) \}$ and $ M_G(G')= M_G(H') \cup \{ (u',v') \}$.
\end{enumerate}

In both cases we have $\Big | |M_G(G')| - |M_G(G)| \Big | \leq 1$. So $|M_G|$ has the Lipschitz property. Note that there is no assumption about the vertices $v,v'$ picked in the first iteration of the Greedy algorithm. So this proof is valid assuming vertices $v,v'$ are of minimum degree, i.e. the proof works for $|M_{KS}|$ and $|M_{OKS}|$ as well.

\section{Proof of Theorem \ref{OKS}}
We prove the theorem in two steps. In Step 1, first we embed the dynamics of input and output degree sequences during the algorithm in a continuous time Markov process for a random network with bounded degrees. Then, we find explicit expressions for differential equations governing the dynamics of degree sequences for infinitely large networks. Then, we show that for a finite random network, the dynamics of degree sequences can be approximated by the solution of the presented initial value problem. Finally, we show that this solution spends zero time in the cases where there is no vertex of degree one so the number of iterations of the algorithm there is no vertex of degree one on the right side of the random network is sublinear w.r.t. the size of the network. In Step 2, we generalize the proof to the unbounded, but finite mean, degree sequences.

{\bf Step 1: Embedding in a continuous time Markov process.} Assume in addition that the asymptotic empirical degree distributions are bounded: $ p_{in}(i)= p_{out}(i) =0 $ for $i >N$ i.e. for every vertex $v \in L \cup R$ we have $\deg(v) \leq N$. First, we embed the dynamics of the algorithm in a continuous time Markov process. To go to continuous time, define $G^n(t)$ as the $n$-vertex network at time $t \in \mathbb R$ where state changes $G^n=G^n - \{ u, v \}$ occur at $\Exp(n)$ interarrival times. More precisely, let $\tau_1, \tau_2, \ldots$ be i.i.d $\Exp(n)$ random variables, i.e. the probability density function is $n e^{- nt}$ for $t \geq 0$ so $\mathbb E(\tau_i)= \frac{1}{n}$. The first state change $G^n=G^n - \{ u, v \}$ occurs at time $t=\tau_1$, the second one occurs at $t=\tau_1+\tau_2$ and so forth. Now we construct a Markov process on $\mathbb R ^{2N}$ which describes the performance of the algorithm. The transition kernel of the Markov process will be described later. Define $X^{(n)}(t) , Y^{(n)}(t) \in \mathbb R ^ N$ as:
$$X^{(n)}_k(t)= \frac{1}{n} \big | \{ v \in R : \deg(v)=k \text{ in } G^n(t) \} \big | , $$
$$Y^{(n)}_k(t)= \frac{1}{n} \big | \{ v \in L : \deg(v)=k \text{ in } G^n(t) \} \big | , $$
for $k=1, 2, \ldots , N$. In addition let $m=m(X^{(n)}(t))$ be the minimum degree of vertices in $R(t)$: $m(X^{(n)}(t))=\min \{ k : X^{(n)}_k(t) \neq 0 \}$ so letting $v_1=v, u_1=u$ whenever a state change occurs we have $\deg(v_1)=m$. Let $(u_i,v) \in E$ for $i=1, \ldots, m$, $K=\deg(u_1)$ and $(u,v_j) \in E$ for $j=1, \ldots, K$. So for a network of size $n$ we have the following conditional degree distributions for vertices $u_1 , \ldots, u_m , v_2 , \ldots , v_K$:
$$ \mathbb P_n \left (\deg(u_i)=k | \mathcal A_{i-1}(t) \right ) = \frac{nkY^{(n)}_k(t) - k \sum \limits_{j=1}^{i-1} 1_{\deg(u_j)=k}}{n\sum \limits_{k=1}^N k Y^{(n)}_k(t) - \sum \limits_{j=1}^{i-1} \deg(u_j)} , $$
$$ \mathbb P_n \left (\deg(v_i)=k | \mathcal B_{i-1}(t) \right ) = \frac{nkX^{(n)}_k(t) - k \sum \limits_{j=1}^{i-1} 1_{\deg(v_j)=k}}{n\sum \limits_{k=1}^N k X^{(n)}_k(t) - \sum \limits_{j=1}^{i-1} \deg(v_j)}  , $$
where $\mathcal A_{i}(t)= (\deg(u_1), \ldots, \deg(u_i), Y^{(n)}(t)), \mathcal B_{i}(t)= (\deg(v_1), \ldots, \deg(v_i), X^{(n)}(t))$. Note that since interarrival times are i.i.d exponential random variables, $X^{(n)}(t) , Y^{(n)}(t)$ are continuous time Markov processes.

Letting $\tilde X^{(n)},\tilde Y^{(n)} \in \mathbb R ^N$ be the corresponding vectors after one state change for $x,y \in \mathbb R ^N$ define functions $\mathbb F^n, \mathbb G^n : \mathbb R ^{2N} \to \mathbb R^N$ as:
$$\mathbb F^n(x,y)=n \mathbb E_n (\tilde X^{(n)}- X^{(n)} | X^{(n)}=x,Y^{(n)}=y) , $$
$$\mathbb G^n(x,y)=n \mathbb E_n (\tilde Y^{(n)}- Y^{(n)} | X^{(n)}=x,Y^{(n)}=y)  , $$
where $\mathbb E_n $ is expected value w.r.t $\mathbb P_n$. Since the process is Markov, probability distribution of $\tilde X^{(n)},\tilde Y^{(n)}$ depends only on $X^{(n)},Y^{(n)}$. 

{\bf Asymptotic initial value problem:} Define $\mathbb P (\deg(u_i)=k) = \frac{kY^{(n)}_k(t)}{\sum \limits_{k=1}^N k Y^{(n)}_k(t)}, \mathbb P(\deg(v_i)=k) = \frac{kX^{(n)}_k(t)}{\sum \limits_{k=1}^N k X^{(n)}_k(t)}$. Note that $\mathbb P_n, \mathbb P$ can be defined for every $x,y \in \mathbb R^N$ with non-negative components. Now, for arbitrary $x,y$, some algebra gives:
\begin{eqnarray} \label{probineq1}
| \mathbb P_n(\deg(u_i)=k| \mathcal A_{i-1}(t)) - \mathbb P(\deg(u_i)=k) | \leq  \frac{C_1}{n} ,
\end{eqnarray}
\begin{eqnarray} \label{probineq2}
| \mathbb P_n(\deg(v_i)=k| \mathcal B_{i-1}(t)) - \mathbb P(\deg(v_i)=k) | \leq   \frac{C_2}{n} .
\end{eqnarray} 
For $C_1=\frac{2N^2}{\sum \limits_{k=1}^N k y_k}, C_2=\frac{2N^2}{\sum \limits_{k=1}^N k x_k}$. Define: 
$$\mathbb F(x,y)=n \mathbb E (\tilde X^{(n)}- X^{(n)} | X^{(n)}=x,Y^{(n)}=y) , $$
$$\mathbb G(x,y)=n \mathbb E (\tilde Y^{(n)}- Y^{(n)} | X^{(n)}=x,Y^{(n)}=y) , $$
where $\mathbb E $ is expected value w.r.t $\mathbb P$. Since
\begin{eqnarray} \label{trans1}
n \tilde X^{(n)}_k=n X^{(n)}_k + \sum \limits_{j=2}^{K} [1_{\deg(v_j)=k+1}-1_{\deg(v_j)=k}]-1_{k=m} ,
\end{eqnarray}
\begin{eqnarray} \label{trans2}
n \tilde Y^{(n)}_k=nY^{(n)}_k + \sum \limits_{j=2}^{m} [1_{\deg(u_j)=k+1}-1_{\deg(u_j)=k}]-1_{k=K} ,
\end{eqnarray}
inequalities \eqref{probineq1}, \eqref{probineq2} yield
\begin{eqnarray} \label{ineq1}
\| \mathbb F^n(x,y)- \mathbb F(x,y) \|_1 \leq \frac{4N^4}{\sum \limits_{k=1}^N k x_k} \frac{1}{n} , 
\end{eqnarray}
\begin{eqnarray} \label{ineq2}
\| \mathbb G^n(x,y)- \mathbb G(x,y) \|_1 \leq \frac{4N^4}{\sum \limits_{k=1}^N k y_k} \frac{1}{n} , 
\end{eqnarray}
and,
$$ \mathbb F(x,y)=(\frac{\|A^2y\|_1}{\|Ay\|_1}-1)\frac{1}{\|Ax\|_1}(SAx-Ax)-1_{m(x)} , $$
$$ \mathbb G(x,y)=-\frac{Ay}{\|Ay\|_1}+\frac{m(x)-1}{\|Ay\|_1}(SAy-Ay) , $$
where $A$ and $S$ are moment matrix and shift matrix respectively, i.e., $A,S \in \mathbb R^{N \times N}$, $A_{ij}$ is $i$ for $i=j$ and $0$ otherwise, and $S_{ij}$ is $1$ for $i=j-1$ and $0$ otherwise, $m(x)=\min \{ k : x_k \neq 0 \}$ and $1_{m} \in \mathbb R^N$ is the vector in which $m$-th component is $1$ and all others are $0$. Note that $\|AX^{(n)}(t)\|_1= \|AY^{(n)}(t)\|_1$ because for finite $n$ always $|E(t)|= n \sum \limits_{k=1}^N k X^{(n)}_k(t)= n \sum \limits_{k=1}^N k Y^{(n)}_k(t)$. Besides, transition kernel of the Markov process can be formulated by $\mathbb P_n$ according to \eqref{trans1}, \eqref{trans2}.

%\item
{\bf Approximating the dynamics of the degree sequences by the solution of asymptotic initial value problem:} Now we can use {\em Kurtz's Theorem} \cite{Kurtz}. Given functions $\mathbb F, \mathbb G: \mathbb R^{2N} \to \mathbb R^N$ and positive constant $T$, define $x(t),y(t): [0,T] \to \mathbb R^N$ as the solution of the initial value problems 
\begin{eqnarray}  
\dot{x} = \mathbb F(x,y), x_k(0)= p_{in}(k) , k=1, \ldots, N \label{diffeq} , \\
\dot{y} = \mathbb G(x,y), y_k(0)= p_{out}(k) , k=1, \ldots, N \label{diffeq2} ,  
\end{eqnarray}
and let $\mathcal E = \{z \in \mathbb R^N \text{such that: } \epsilon < \sum \limits_{k=1}^N kz_k \leq N \}$. Suppose the following statements hold:
\begin{enumerate}
\item 
$\lim \limits_{n \to \infty} \sup \limits_{z_1 , z_2 \in \mathcal E} \| \mathbb F^n(z_1,z_2) - \mathbb F (z_1,z_2) \|_1 =0 . $
\item 
$\lim \limits_{n \to \infty} \sup \limits_{z_1 , z_2 \in \mathcal E} \| \mathbb G^n(z_1,z_2) - \mathbb G (z_1,z_2) \|_1 =0 . $
\item
for all $k=1, \ldots, N$, $\lim \limits_{n \to \infty} X^{(n)}_k(0) =p_{in}(k) . $
\item
for all $k=1, \ldots, N$, $\lim \limits_{n \to \infty} Y^{(n)}_k(0) =p_{out}(k) . $
\item
functions $\mathbb F, \mathbb G$ are Lipschitz {(in the classic sense)}. 
\end{enumerate}
Then 
\begin{eqnarray} \label{closeness1}
\lim \limits_{n \to \infty} \mathbb P_n \left ( \exists t \in [0,T]: m(X^{(n)}(t)) \neq m(x(t)) \right ) =0 .
\end{eqnarray}
Letting $T= T(\epsilon) = \sup \{ t : \sum \limits_{k=1}^N kx_k(t) > \epsilon, \sum \limits_{k=1}^N ky_k(t) > \epsilon\}$ for some arbitrary $\epsilon >0$, the first two conditions are satisfied by \eqref{ineq1}, \eqref{ineq2}. By the definition of asymptotic empirical degree distributions $\lim \limits_{n \to \infty} X^{(n)}_k(0) =p_{in}(k)$ and $\lim \limits_{n \to \infty} Y^{(n)}_k(0) =p_{out}(k)$. On the other hand, the initial value problems \eqref{diffeq}, \eqref{diffeq2} have unique solutions since defining metric $d$ on $\mathbb R ^N$ as $d(x,y)=\|x-y\|_1 + 1_{m(x) \neq m(y)}$, $\mathbb F, \mathbb G$ are Lipschitz with respect to this metric i.e. there is a $B < \infty $ such that for all $x,x',y,y' \in \mathbb R^N$
$$ d(\mathbb F(x,y),\mathbb F(x',y')) < B (d(x,x')+d(y,y')) , $$
$$ d(\mathbb G(x,y),\mathbb G(x',y')) < B (d(x,x')+d(y,y')) . $$
Note that stopping time at $T(\epsilon)$, i.e. when $\epsilon n$ edges are {remaining in the network to be removed by the algorithm, will not cause any problem since} continuing the algorithm from that point on cannot add more than $\epsilon$ edges to the matching on a scale relative to the number $n$ of vertices.

%\item
{\bf Properties of the asymptotic initial value problems:} The useful fact about the solutions of \eqref{diffeq}, \eqref{diffeq2} is that Lebesgue measure of the set $\{ 0<t<T: m(x(t)) >1 \}$ is zero. Suppose it is not. So there are $0< t_1< t_2$ such that $x_1(t)=0$ for all $t \in [t_1, t_2]$ so $\frac{dx_1(t)}{dt}=0$ for all $t \in (t_1, t_2)$. But
$$\frac{dx_1(t)}{dt} = \mathbb F_1(x,y)= (\frac{\|A^2y\|_1}{\|Ay\|_1}-1)\frac{1}{\|Ax\|_1}(2x_2 (t)-x_1(t))$$
implies $x_2(t)=0$ for all $t \in (t_1, t_2)$ which means $m(x(t)) >2$ for all $t \in (t_1, t_2)$. Repeating this argument for $x_2$ now we will get $m(x(t)) >3$ and so on. Therefore, $m(x(t)) >  N$ which is impossible. Thus, the set $\{ 0<t<T: m(x(t)) >1 \}$ has zero Lebesgue measure.

%\item
{\bf Sublinearity of the number of iterations of the algorithm with no degree one vertex:} Let $J^{(n)} \subset \{ 1,2, \ldots , n \}$ be the set of indices $i$ of iterations of $OKS$ such that after the $i$-th iteration the minimum degree on the right side of the network is larger than one. Since the set $\{ 0<t<T: m(x(t)) >1 \}$ has zero Lebesgue measure by \eqref{closeness1} Lebesgue measure of the set $\{ 0<t<T: m(X^{(n)}(t)) >1 \}$ goes to 0 as $n$ grows. Because the Lebesgue measure of the set $\{ 0<t<T: m(X^{(n)}(t)) >1 \}$ is $ \sum \limits_{i \in J^{(n)}}\tau_{i+1}$ we have $\lim \limits_{n \to \infty} \sum \limits_{i \in J^{(n)}}\tau_{i+1} =0$ but by the Law of Large Numbers $\lim \limits_{n \to \infty} \frac{1}{|J^{(n)}|}\sum \limits_{i \in J^{(n)}} n \tau_{i+1} = \mathbb E (n \tau_1) =1$. Therefore $ \lim \limits_{n \to \infty} \frac{|J^{(n)}|}{n} = \lim \limits_{n \to \infty} \frac{|J^{(n)}|}{n} \frac{1}{|J^{(n)}|}\sum \limits_{i \in J^{(n)}} n \tau_{i+1} =\lim \limits_{n \to \infty} \sum \limits_{i \in J^{(n)}}\tau_{i+1}=0 $ i.e. the number of iterations of $OKS$ algorithm for which there is no vertex of degree one on the right side of the network is sublinear w.r.t. the size of the network. 

%\item
{\bf Sublinearity of difference between the output of $OKS$ and maximum matching:} Using $\lim \limits_{n \to \infty} \frac{|J^{(n)}|}{n}=0$ we prove that the size of the matching provided by the $OKS$ algorithm is away from maximum matching size by a sublinear factor. Starting the algorithm, as long as the minimum degree on the right side of the network is one, $OKS$ makes no mistake, i.e. the size of the matching by $OKS$ is the same as the size of maximum matching. When the minimum degree is $m=m(X^{(n)}(t)) > 1$ it is possible that $OKS$ picks a vertex on the left side which is not the optimal choice. We make it optimal by manipulating the network: if $v \in R, u \in L, \deg(v)=m$ are the chosen vertices in the iteration of the algorithm to be removed from the network, $M_{OKS}= M_{OKS} \cup \{ (u,v) \}$, manipulate the network by removing all other $m-1$ edges connected to $v$. Since $|M^*|$ has the Lipschitz property, removing these $m-1 $ edges will change the size of the maximum matching by at most $m-1$. Since $m$ is the minimum degree and the average degree is bounded, $m-1$ is bounded as well. On the other hand, the number of iterations that $OKS$ will face such cases is sublinear w.r.t. the size of the network, so the whole number of possible errors, or in other words, the whole deviation from maximum matching made by $OKS$ is sublinear, i.e. $\lim \limits_{n \to \infty} \frac{| M_{OKS}(G) |}{n}= \lim \limits_{n \to \infty} \frac{| M^*(G) |}{n}$.
%\end{itemize} 

{\bf Step 2: Generalization to unbounded degree.} Now to generalize the proof to cases where the asymptotic empirical degree distributions are not bounded, we use the classical technique of truncation. For arbitrary $\epsilon > 0$, let $N$ be large enough such that $\sum \limits_{k=N+1}^ \infty k p_{in}(k) < \frac{\epsilon}{2},  \sum \limits_{i=N+1}^ \infty k p_{out}(k) < \frac{\epsilon}{2} $. In random network $G$ remove some edges in order to have no vertex of degree larger than $N$ to get random network $H$ which has bounded asymptotic empirical degree distributions. By Step 1,
\begin{eqnarray}\label{step21}
\lim \limits_{n \to \infty} \frac{| M_{OKS}(H) |}{n}= \lim \limits_{n \to \infty} \frac{| M^*(H) |}{n} .
\end{eqnarray}
Because by Lemma \ref{LipProofHeuristics} both functions $|M^*|,|M_{OKS}|$ have the Lipschitz property and asymptotically the number of edges removed from $G$ to get $H$ is less than $n \epsilon$: 
\begin{eqnarray} \label{step22}
\lim \limits_{n \to \infty} \frac{\big| |M^*(H)|-|M^*(G)| \big |}{n} <  \epsilon ,
\end{eqnarray}
\begin{eqnarray} \label{step23}
\lim \limits_{n \to \infty} \frac{\big | |M_{OKS}(H)|-|M_{OKS}(G)| \big |}{n} <  \epsilon .
\end{eqnarray}
Now \eqref{step21}, \eqref{step22}, \eqref{step23} imply the desired result. Further, when $\epsilon \to 0$, $N \to \infty$, so formally we can take $N=\infty$ and write the functions $\mathbb F, \mathbb G : \mathbb R^\infty \to \mathbb R^\infty$ as

\begin{eqnarray*} 
\mathbb F(x,y) &=& (\frac{\|A^2y\|_1}{\|Ay\|_1}-1)\frac{1}{\|Ax\|_1}(SAx-Ax)-1_{m(x)} ,\\
\mathbb G(x,y) &=& -\frac{Ay}{\|Ay\|_1}+\frac{m(x)-1}{\|Ay\|_1}(SAy-Ay) ,
\end{eqnarray*}
for matrices $A,S \in \mathbb R^{\infty \times \infty}$ provided $\| A^2 x(0) \|_1 = \sum \limits_{k=1}^ \infty k^2 p_{in}(k) < \infty$ or $\| A^2 y(0) \|_1 = \sum \limits_{k=1}^ \infty k^2 p_{out}(k) < \infty$. 

\section{Proof of Theorems \ref{greedymean},\ref{KSgeneral}}
As we saw in the proof of Theorem \ref{OKS}, The asymptotic performance of Greedy, $OKS$ and $KS$ algorithms can be described by asymptotic empirical degree distributions which are solutions of some initial value problems. If we find functions $\mathbb F, \mathbb G : \mathbb R^\infty \to \mathbb R^\infty$ for Greedy and $KS$ we will have
\begin{eqnarray*} 
\mathbb F(x,y) & =&  (\frac{\|A^2y\|_1}{\|Ay\|_1}-1)\frac{1}{\|Ax\|_1}(SAx-Ax)- \frac{x}{\|x\|_1} , \\
 \mathbb G(x,y) &=& -\frac{Ay}{\|Ay\|_1}++\frac{ \frac{\|Ax\|_1}{\|x\|_1}-1}{\|Ay\|_1}(SAy-Ay)
\end{eqnarray*}
for Greedy and for $KS$
\begin{eqnarray*} \label{KSfunctions1}
\mathbb F(x,y) &=& \frac{x_m}{x_m + y_m} \left [ (\frac{\|A^2y\|_1}{\|Ay\|_1}-1)\frac{1}{\|Ax\|_1}(SAx-Ax)-1_{m} \right ] 
+ \frac{y_m}{x_m + y_m} \left [-\frac{Ax}{\|Ax\|_1}+\frac{m-1}{\|Ax\|_1}(SAx-Ax) \right ]
\end{eqnarray*}
\begin{eqnarray*} \label{KSfunctions2}
\mathbb G(x,y) &=& \frac{y_m}{x_m + y_m} \left [ (\frac{\|A^2x\|_1}{\|Ax\|_1}-1)\frac{1}{\|Ay\|_1}(SAy-Ay)-1_{m} \right ] 
+ \frac{x_m}{x_m + y_m} \left [-\frac{Ay}{\|Ay\|_1}+\frac{m-1}{\|Ay\|_1}(SAy-Ay) \right ]
\end{eqnarray*}
where $m=\min \{ m(x(t)), m(y(t)) \}$ is the minimum degree. Since for $KS$ we have $\mathbb G(x,y) = \mathbb F(y,x)$, when $x(0)=y(0)$ (i.e. $p_{in}=p_{out}$) $x(t)=y(t)$ and
$$ \dot x = \mathbb F(x)= (\frac{\|A^2x\|_1}{\|Ax\|_1}+m-2)  \frac{SAx-Ax}{2\|Ax\|_1} - \frac{Ax}{2\|Ax\|_1} - \frac{1_m}{2}$$
Thus for any degree distribution $p$ the dynamics of $KS$ is the same for both $DD(p)$ and $DD(p,p)$. Further, regarding the results provided in Theorems \ref{greedymean}, \ref{KSgeneral} the relative size of the output of the algorithm as well as the dynamics of the algorithm is the same for all random networks $ER(\lambda)$, $UFS(\lambda)$, $DD(p)$ and $DD(p,p)$ where probability distribution $p$ is $\Poisson(\lambda)$ because they all are sharing the asymptotic empirical degree distributions. So Theorem \ref{greedy} implies $\lim \limits_{n \to \infty }\frac{|M_G(G)|}{n} = 1 - \frac{ \log (2 - e^{- \lambda})}{\lambda}$. For $KS$, all mentioned statements are proved for undirected $ER(\lambda)$ by Karp and Sipser \cite{KS} so are valid for the desired class of random networks.

\section{Proof of Theorem \ref{MMSize}}
Here we assume in addition that $\| A^2 x(0) \|_1 = \sum \limits_{k=1}^ \infty k^2 p_{in}(k) < \infty$ and $\| A^2 y(0) \|_1 = \sum \limits_{k=1}^ \infty k^2 p_{out}(k) < \infty$. Generalization to the case where above quantities are not bounded is straightforward similar to Step 2 in the proof of Theorem \ref{OKS} and is omitted.

Because $\lim \limits_{n \to \infty} \frac{| M^*(G)|}{n}=\lim \limits_{n \to \infty} \frac{| M_{KS}(G)|}{n}$ it suffices to show $U^*=\lim \limits_{n \to \infty} 1-\frac{| M_{KS}(G)|}{n}$. To show the latter claim, we run $KS$ algorithm and find the number of vertices left unmatched by the algorithm. Similar to what we did in the proof of Theorem \ref{OKS} the asymptotic fraction of vertices left unmatched by $KS$ is
\begin{eqnarray*}
1- \lim \limits_{n \to \infty} \frac{| M_{KS}(G)|}{n}=p_{in}(0)  + \int \limits_0 ^ T  \left [ \frac{x_m}{x_m + y_m}(\frac{\|A^2y\|_1}{\|Ay\|_1}-1) + \frac{(m-1)y_m}{x_m + y_m}\right ] \frac{x_1(t)}{\|Ax\|_1} dt
\end{eqnarray*}
where 
$$T= \sup \{ t : \sum \limits_{k=1}^\infty kx_k(t) > 0, \sum \limits_{k=1}^\infty ky_k(t) > 0\}$$ 
$$ \dot{x} =\mathbb F(x,y), x(0)=p_{in}$$
 $$ \dot{y} =\mathbb F(y,x), y(0)=p_{out}$$ 
\begin{eqnarray*} \label{KSDynamic}
\mathbb F(x,y)  = \frac{x_m}{x_m + y_m} \left [ (\frac{\|A^2y\|_1}{\|Ay\|_1}-1)\frac{1}{\|Ax\|_1}(SAx-Ax)-1_{m} \right ] 
+ \frac{y_m}{x_m + y_m} \left [-\frac{Ax}{\|Ax\|_1}+\frac{m-1}{\|Ax\|_1}(SAx-Ax) \right ]
\end{eqnarray*}
and $m=\min \{ m(x(t)), m(y(t)) \}$ is the minimum degree. Since as we saw in the proof of Theorem \ref{OKS} the set $\{ t: m(x(t)) > 1 , m(y(t))>1 \}$ is of zero Lebesgue measure without loss of generality in all integrations we can assume $m=1$, especially
\begin{eqnarray*}
1- \lim \limits_{n \to \infty} \frac{| M_{KS}(G)|}{n} =p_{in}(0)
+ \int \limits_0 ^ T  \frac{x_1}{x_1 + y_1}(\frac{\|A^2y\|_1}{\|Ay\|_1}-1)  \frac{x_1(t)}{\|Ax\|_1} dt 
\end{eqnarray*}
Now define:
$$ x_0(t)= p_{in}(0)+\int \limits_0^t \frac{x_1}{x_1 + y_1}(\frac{\|A^2y(s)\|_1}{\| Ay(s) \|_1}-1) \frac{x_1(s)}{\|Ax\|_1}ds , $$
$$ y_0(t)= p_{out}(0)+\int \limits_0^t \frac{y_1}{x_1 + y_1}(\frac{\|A^2x(s)\|_1}{\|Ax(s)\|_1}-1) \frac{y_1(s)}{\|Ay\|_1}ds , $$
$$\mu (t)= \sum \limits_{i=0}^ \infty i x_i(t)= \sum \limits_{i=0}^ \infty i y_i(t)=\|Ax\|_1=\|Ay\|_1 , $$
$$ \Phi_{in}(t,u)=\sum \limits_{i=0}^ \infty x_i(t)u^i , $$
$$ \Phi_{out}(t,u)=\sum \limits_{i=0}^ \infty y_i(t)u^i , $$
$$\phi_{in}(t,u)= \sum \limits_{i=1}^ \infty \frac{ix_i(t)}{\| Ax(t) \|_1}u^{i-1} , $$
$$\phi_{out}(t,u)=\sum \limits_{i=1}^ \infty \frac{iy_i(t)}{\| Ay(t) \|_1}u^{i-1} , $$
$$ \phi_{in}(t,w_2(t))=w_3(t) , $$
$$ \phi_{in}(t,1-w_1(t))=1-w_4(t) , $$
$$ \phi_{out}(t,w_4(t))=w_1(t) , $$
$$ \phi_{out}(t,1-w_3(t))=1-w_2(t) , $$
$$ V(t)= \sum \limits_{i=1}^\infty x_i(t) = \|x(t)\|_1 , $$
\begin{eqnarray*} 
U(t)=\frac{1}{2} & \Big  [ & \Phi_{in}(t,w_2(t))+\Phi_{in}(t,1-w_1(t)) +\Phi_{out}(t,w_4(t))  +\Phi_{out}(t,1-w_3(t)) -2 \\
& +& \mu (t) \big [w_3(t) (1-w_2(t))+w_1(t) (1-w_4(t)) \big ] \Big ]
\end{eqnarray*}
Since $x_1(T)=x_2(T)=\ldots = 0, y_1(T)=y_2(T)=\ldots = 0$ we have $V(T)=0,\mu (T)=0$. But $V(0)=1-p_{in}(0)$ and for $m=1$
\begin{eqnarray*} 
\dot V(t)=  \frac{d \|x \|_1}{dt}
&=& \frac{x_1}{x_1 + y_1} \left [-(\frac{\|A^2y\|_1}{\|Ay\|_1}-1)\frac{x_1}{\|Ax\|_1}-1 \right ] - \frac{y_1}{x_1 + y_1} \\
&=& -\frac{x_1}{x_1 + y_1}(\frac{\|A^2y\|_1}{\|Ay\|_1}-1)\frac{x_1}{\|Ax\|_1}-1 \\
& =& -\dot x_0(t) -1
\end{eqnarray*}
i.e. $-1+p_{in}(0)=V(T)-V(0)=x_0(0)- x_0(T)-T=p_{in}(0)- x_0(T)-T$. Therefore $\Phi_{in}(T,u)=x_0(T), \Phi_{out}(T,u)=y_0(T), U(T)=1-2T$ and
\begin{eqnarray} \label{KSUnmatched}
1 - \lim \limits_{n \to \infty} \frac{| M_{KS}(G)|}{n}=x_0(T)=y_0(T)=1-T
\end{eqnarray}
On the other hand, as long as $m=1$:
\begin{eqnarray*} 
\dot \mu (t) = \frac{d \mu}{dt} &=&  \frac{y_1}{x_1 + y_1} \left [-\frac{\| A^2x \|_1}{\|Ax\|_1} \right ] 
+ \frac{x_1}{x_1 + y_1} \left [ (\frac{\|A^2y\|_1}{\|Ay\|_1}-1)\frac{-x_1-\|Ax\|_1+x_1 }{\|Ax\|_1} -1 \right ] \\
&=& -\frac{x_1}{x_1 + y_1}\frac{\|A^2y\|_1}{\|Ay\|_1} - \frac{y_1}{x_1 + y_1} \frac{\|A^2x\|_1}{\|Ax\|_1} \\
&=& -\frac{x_1 \|A^2y\|_1 + y_1 \|A^2x \|_1}{(x_1 + y_1)\|Ax\|_1} ,
\end{eqnarray*}
\begin{eqnarray*} 
\frac{d}{dt}\Phi_{in}(t,w_2(t))
&=& \frac{x_1}{x_1 + y_1}\Big [ (\frac{\|A^2y\|_1}{\|Ay\|_1}-1)\phi_{in}(t,w_2(t))(1-w_2(t)) 
-w_2(t) \Big] \\
 &+& \frac{y_1}{x_1 + y_1} \big [-w_2(t) \phi_{in}(t,w_2(t)) \big ] 
+ \|Ax\|_1 \dot w_2(t) \phi_{in}(t,w_2(t))\\
&=& \frac{x_1}{x_1 + y_1} \left [ (\frac{\|A^2y\|_1}{\|Ay\|_1}-1)w_3(t)(1-w_2(t))-w_2(t) \right ] \\
&+& \frac{y_1}{x_1 + y_1} \big [-w_2(t) w_3(t) \big ] 
+ \|Ax\|_1 \dot w_2(t) \phi_{in}(t,w_2(t)) ,
\end{eqnarray*}
\begin{eqnarray*} 
\frac{d}{dt}\Phi_{in}(t,1-w_1(t)) 
&=&\frac{x_1}{x_1 + y_1} \Big [(\frac{\|A^2y\|_1}{\|Ay\|_1}-1)\phi_{in}(t,1-w_1(t))w_1(t)  
 -1+w_1(t) \Big ] \\
&+& \frac{y_1}{x_1 + y_1} \big [-1+w_1(t) \phi_{in}(t,1-w_1(t)) \big ] 
- \|Ax\|_1 \dot w_1(t) \phi_{in}(t,1-w_1(t)) \\
&=& \frac{x_1}{x_1 + y_1} \left [(\frac{\|A^2y\|_1}{\|Ay\|_1}-1)w_1(t)(1-w_4(t))-1+w_1(t) \right ] \\
&+& \frac{y_1}{x_1 + y_1} \big [-(1-w_1(t)) (1-w_4(t)) \big ] 
- \|Ax\|_1 \dot w_1(t) \phi_{in}(t,1-w_1(t))  ,
\end{eqnarray*}
\begin{eqnarray*} 
\frac{d}{dt} \left [ \mu (t) w_3(t) (1-w_2(t)) \right ] 
&=& \dot \mu (t) w_3(t) (1-w_2(t)) 
+\mu (t) \dot w_3(t) (1-w_2(t))-\mu (t) w_3(t) \dot w_2(t) \\
&=& \left (-\frac{x_1\|A^2y\|_1 + y_1 \|A^2x\|_1}{(x_1 + y_1)\|Ax\|_1} \right ) w_3(t) (1-w_2(t)) 
+\|Ax\|_1 \big [ \dot w_3(t) (1-w_2(t))- w_3(t) \dot w_2(t) \big ] .\\
\end{eqnarray*}
The above equations imply:
\begin{eqnarray*} 
&& 2\frac{dU(t)}{dt} = \\
&& \frac{x_1}{x_1 + y_1}\left [(\frac{\|A^2y\|_1}{\|Ay\|_1}-1)w_3(t)(1-w_2(t))-w_2(t)\right ] 
+ \frac{y_1}{x_1 + y_1} \left [-w_2(t) w_3(t)\right ] + \|Ax\|_1 \dot w_2(t) \phi_{in}(t,w_2(t)) \\
&+& \frac{x_1}{x_1 + y_1}\left [(\frac{\|A^2y\|_1}{\|Ay\|_1}-1)w_1(t)(1-w_4(t))-1+w_1(t)\right ] 
+ \frac{y_1}{x_1 + y_1} \left [-(1-w_1(t)) (1-w_4(t))\right ] - \|Ax\|_1 \dot w_1(t) \phi_{in}(t,1-w_1(t)) \\
&+& \frac{y_1}{x_1 + y_1}\left [(\frac{\|A^2x\|_1}{\|Ax\|_1}-1)w_1(t)(1-w_4(t))-w_4(t)\right ] 
+ \frac{x_1}{x_1 + y_1} \left [-w_4(t) w_1(t)\right ] + \|Ay\|_1 \dot w_4(t) \phi_{in}(t,w_4(t)) \\
&+& \frac{y_1}{x_1 + y_1}\left [(\frac{\|A^2x\|_1}{\|Ax\|_1}-1)w_3(t)(1-w_2(t))-1+w_3(t)\right ] 
+ \frac{x_1}{x_1 + y_1} \left [-(1-w_3(t)) (1-w_2(t))\right ] - \|Ay\|_1 \dot w_3(t) \phi_{in}(t,1-w_3(t)) \\
&+& (-\frac{x_1\|A^2y\|_1 + y_1 \|A^2x\|_1}{(x_1 + y_1)\|Ax\|_1}) w_3(t) (1-w_2(t)) 
+\|Ax\|_1 \left [\dot w_3(t) (1-w_2(t))- w_3(t) \dot w_2(t)\right ] \\
&+& (-\frac{x_1\|A^2y\|_1 + y_1 \|A^2x\|_1}{(x_1 + y_1)\|Ax\|_1}) w_1(t) (1-w_4(t)) 
+\|Ax\|_1 \left [\dot w_1(t) (1-w_4(t))- w_1(t) \dot w_4(t)\right ] \\
\end{eqnarray*}
i.e.
\begin{eqnarray*} 
2\frac{d}{dt}U(t)
&=& \frac{x_1}{x_1 + y_1}\left [-\frac{\|A^2y\|_1}{\|Ay\|_1}(w_3(t) (1-w_2(t)) \right . 
\left .+w_1(t) (1-w_4(t)))  \right . 
+ \left . (\frac{\|A^2y\|_1}{\|Ay\|_1}-1)w_3(t)(1-w_2(t))-w_2(t) \right . \\
&& + \left . (\frac{\|A^2y\|_1}{\|Ay\|_1}-1)w_1(t)(1-w_4(t))-1+w_1(t)  \right . 
\left . -w_4(t) w_1(t)-(1-w_3(t)) (1-w_2(t)) \vphantom{\frac{\|A^2y\|_1}{\|Ay\|_1}}\right ] \\
&+& \frac{y_1}{x_1 + y_1} \left [-\frac{\|A^2x\|_1}{\|Ax\|_1}(w_3(t) (1-w_2(t))  \right . 
\left . +w_1(t) (1-w_4(t)))  \right . 
\left . -w_2(t) w_3(t)-(1-w_1(t)) (1-w_4(t)) \right . \\
&& \left . +(\frac{\|A^2x\|_1}{\|Ax\|_1}-1)w_1(t)(1-w_4(t))-w_4(t)  \right . 
\left . +(\frac{\|A^2x\|_1}{\|Ax\|_1}-1)w_3(t)(1-w_2(t))-1+w_3(t)\right ] \\
&+& \|Ax\|_1 \left [\dot w_2(t) w_3(t) - \dot w_1(t) (1-w_4(t))  \right . 
\left . + \dot w_4(t) w_1(t) - \dot w_3(t) (1-w_2(t)) \right . \\
&+& \left .  \dot w_3(t) (1-w_2(t))- w_3(t) \dot w_2(t)  \right . 
\left . +\dot w_1(t) (1-w_4(t))- w_1(t) \dot w_4(t)\right ] .
\end{eqnarray*}
Simplifying 
\begin{eqnarray*}
2 \frac{d}{dt}U(t)
&=& \frac{x_1}{x_1 + y_1} \left [  -\frac{\|A^2y\|_1}{\|Ay\|_1}(w_3(t) (1-w_2(t))+w_1(t) (1-w_4(t))) \right .
\left . +\frac{\|A^2y\|_1}{\|Ay\|_1}w_3(t)(1-w_2(t)) - w_3(t)(1-w_2(t)) \right . \\
&& \left . - w_2(t)+\frac{\|A^2y\|_1}{\|Ay\|_1}w_1(t)(1-w_4(t))-w_1(t)(1-w_4(t))-1 \right . 
\left . +w_1(t)-w_4(t) w_1(t)-(1-w_3(t)) (1-w_2(t)) \vphantom{\frac{\|A^2y\|_1}{\|Ay\|_1}} \right ]  \\
&+& \frac{y_1}{x_1 + y_1} \left [-\frac{\|A^2x\|_1}{\|Ax\|_1}(w_3(t) (1-w_2(t))  \right . 
\left . +w_1(t) (1-w_4(t))) -w_2(t) w_3(t)-(1-w_1(t)) (1-w_4(t))  \right . \\
&& \left . +\frac{\|A^2x\|_1}{\|Ax\|_1}w_1(t)(1-w_4(t))-w_1(t)(1-w_4(t))-w_4(t)  \right . 
\left . +\frac{\|A^2x\|_1}{\|Ax\|_1}w_3(t)(1-w_2(t))-w_3(t)(1-w_2(t))-1+w_3(t)\right ] 
\end{eqnarray*}
yields
\begin{eqnarray*} 
2 \frac{d}{dt}U(t)= \frac{x_1}{x_1 + y_1}\left [-2\right ] + \frac{y_1}{x_1 + y_1} \left [ -2\right ] =-2
\end{eqnarray*}
Again, since the set $\{ t: m > 0 \}$ is of zero Lebesgue measure integrating both sides of $\frac{d}{dt}U(t)=-1$ we have:
\begin{eqnarray*} 
1-2T=U(T)=U(0)-T.
\end{eqnarray*}
So $U^*=U(0)=1-T$ which is the desired result by \eqref{KSUnmatched}.